\documentclass[journal]{IEEEtran}

\usepackage[cmex10]{amsmath}  
\usepackage{amssymb}
\usepackage{graphicx,color}   
\usepackage[caption=false,font=footnotesize]{subfig} 

\usepackage{ifxetex}
\ifxetex
\usepackage{xunicode}
\else
\usepackage[utf8]{inputenc}
\usepackage[T1]{fontenc}
\usepackage{microtype}
\fi

\usepackage{cite}  

\newtheorem{theorem}{Theorem}
\newtheorem{remark}{Remark}
\newtheorem{lemma}{Lemma}
\newtheorem{definition}{Definition}

\newtheorem{corollary}{Corollary}

\begin{document}

\title{On the Capacity of the Binary-Symmetric Parallel-Relay Network}

\author{Lawrence Ong, Sarah J.\ Johnson, and Christopher M.\ Kellett
\thanks{The authors are with the School of Electrical Engineering and
Computer Science, The University of Newcastle, Callaghan, NSW 2308,
Australia (email: lawrence.ong@cantab.net; \{sarah.johnson,chris.kellett\}@newcastle.edu.au).} }

\maketitle

\begin{abstract}
We investigate the binary-symmetric parallel-relay network where there is one source, one destination, and multiple relays in parallel. We show that forwarding relays, where the relays merely transmit their received signals, achieve the capacity in two ways: with coded transmission at the source and a finite number of relays, or uncoded transmission at the source and a sufficiently large number of relays. On the other hand, decoding relays, where the relays decode the source message, re-encode, and forward it to the destination, achieve the capacity when the number of relays is small. In addition, we show that any coding scheme that requires decoding at any relay is suboptimal in large parallel-relay networks, where forwarding relays achieve strictly higher rates. 
\end{abstract}



\maketitle

\section{Introduction}\label{sec:introduction}

We consider a class of parallel-relay networks where a source sends its data to a destination via many relays.
From a theoretical point of view, analyses of parallel-relay networks are of great interest as they are embedded in more general multiterminal networks. The interest in using relays in a network is also driven by practical applications in which direct communication from a source to the destination is difficult. Parallel-relay networks model scenarios in which the destination obtains data from the source through spatial diversity (e.g., see macroscopic diversity~\cite{bernhardt87}). An example is wired networks where there are multiple routes from the source to the destination.

The network model we consider in this paper is the binary-symmetric parallel-relay (BSPR) network  with $K$ relays as depicted in Fig.~\ref{fig:parallel-relay-network}, where 
each source-to-relay and relay-to-destination channel is a binary-symmetric channel. Our aim is to determine which, if any, transmission schemes can achieve capacity on the binary-symmetric parallel-relay network, and under what conditions. We investigate four possible transmission schemes, which are the different combinations of whether the source and the relays send coded or uncoded signals, plus a hybrid scheme. We obtain capacity results when the network is small (in terms of the number of relays) using \emph{decoding} relays, and when the network is medium to large using \emph{forwarding} relays. The terms \emph{small}, \emph{medium}, and \emph{large} will be made precise in the subsequent sections.

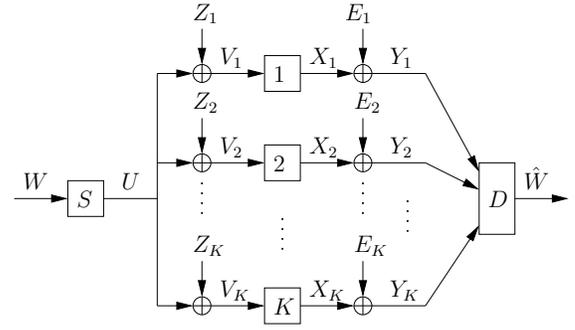
\begin{figure}[t]
\centering
\resizebox{7.5cm}{!}{ 
\begin{picture}(0,0)%
\includegraphics{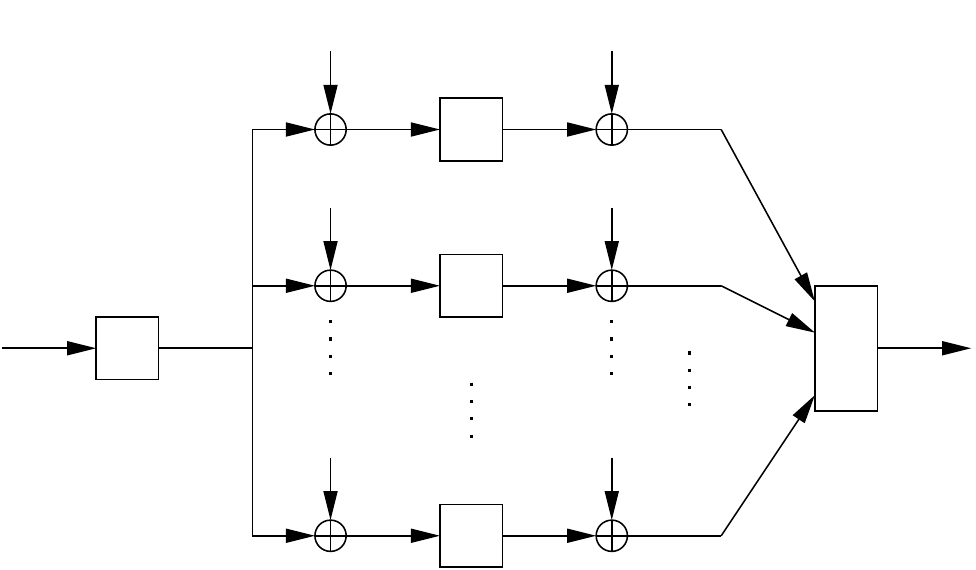}%
\end{picture}%
\setlength{\unitlength}{3947sp}%
\begingroup\makeatletter\ifx\SetFigFont\undefined%
\gdef\SetFigFont#1#2#3#4#5{%
  \fontsize{#1}{#2pt}%
  \fontfamily{#3}\fontseries{#4}\fontshape{#5}%
  \selectfont}%
\fi\endgroup%
\begin{picture}(4674,2721)(-11,-2473)
\put(2176,-1186){\makebox(0,0)[lb]{\smash{{\SetFigFont{12}{14.4}{\familydefault}{\mddefault}{\updefault}{\color[rgb]{0,0,0}$2$}%
}}}}
\put(2176,-436){\makebox(0,0)[lb]{\smash{{\SetFigFont{12}{14.4}{\familydefault}{\mddefault}{\updefault}{\color[rgb]{0,0,0}$1$}%
}}}}
\put(1501,-661){\makebox(0,0)[lb]{\smash{{\SetFigFont{12}{14.4}{\familydefault}{\mddefault}{\updefault}{\color[rgb]{0,0,0}$Z_2$}%
}}}}
\put(1501, 89){\makebox(0,0)[lb]{\smash{{\SetFigFont{12}{14.4}{\familydefault}{\mddefault}{\updefault}{\color[rgb]{0,0,0}$Z_1$}%
}}}}
\put(526,-1486){\makebox(0,0)[lb]{\smash{{\SetFigFont{12}{14.4}{\familydefault}{\mddefault}{\updefault}{\color[rgb]{0,0,0}$S$}%
}}}}
\put(2176,-2386){\makebox(0,0)[lb]{\smash{{\SetFigFont{12}{14.4}{\familydefault}{\mddefault}{\updefault}{\color[rgb]{0,0,0}$K$}%
}}}}
\put(2776, 89){\makebox(0,0)[lb]{\smash{{\SetFigFont{12}{14.4}{\familydefault}{\mddefault}{\updefault}{\color[rgb]{0,0,0}$E_1$}%
}}}}
\put(2851,-661){\makebox(0,0)[lb]{\smash{{\SetFigFont{12}{14.4}{\familydefault}{\mddefault}{\updefault}{\color[rgb]{0,0,0}$E_2$}%
}}}}
\put(3976,-1486){\makebox(0,0)[lb]{\smash{{\SetFigFont{12}{14.4}{\familydefault}{\mddefault}{\updefault}{\color[rgb]{0,0,0}$D$}%
}}}}
\put(1501,-1861){\makebox(0,0)[lb]{\smash{{\SetFigFont{12}{14.4}{\familydefault}{\mddefault}{\updefault}{\color[rgb]{0,0,0}$Z_K$}%
}}}}
\put(2851,-1861){\makebox(0,0)[lb]{\smash{{\SetFigFont{12}{14.4}{\familydefault}{\mddefault}{\updefault}{\color[rgb]{0,0,0}$E_K$}%
}}}}
\put(1726,-286){\makebox(0,0)[lb]{\smash{{\SetFigFont{12}{14.4}{\familydefault}{\mddefault}{\updefault}{\color[rgb]{0,0,0}$V_1$}%
}}}}
\put(1726,-1036){\makebox(0,0)[lb]{\smash{{\SetFigFont{12}{14.4}{\familydefault}{\mddefault}{\updefault}{\color[rgb]{0,0,0}$V_2$}%
}}}}
\put(1726,-2236){\makebox(0,0)[lb]{\smash{{\SetFigFont{12}{14.4}{\familydefault}{\mddefault}{\updefault}{\color[rgb]{0,0,0}$V_K$}%
}}}}
\put(2476,-286){\makebox(0,0)[lb]{\smash{{\SetFigFont{12}{14.4}{\familydefault}{\mddefault}{\updefault}{\color[rgb]{0,0,0}$X_1$}%
}}}}
\put(2476,-1036){\makebox(0,0)[lb]{\smash{{\SetFigFont{12}{14.4}{\familydefault}{\mddefault}{\updefault}{\color[rgb]{0,0,0}$X_2$}%
}}}}
\put(2476,-2236){\makebox(0,0)[lb]{\smash{{\SetFigFont{12}{14.4}{\familydefault}{\mddefault}{\updefault}{\color[rgb]{0,0,0}$X_K$}%
}}}}
\put( 76,-1336){\makebox(0,0)[lb]{\smash{{\SetFigFont{12}{14.4}{\familydefault}{\mddefault}{\updefault}{\color[rgb]{0,0,0}$W$}%
}}}}
\put(4276,-1336){\makebox(0,0)[lb]{\smash{{\SetFigFont{12}{14.4}{\familydefault}{\mddefault}{\updefault}{\color[rgb]{0,0,0}$\hat{W}$}%
}}}}
\put(3151,-286){\makebox(0,0)[lb]{\smash{{\SetFigFont{12}{14.4}{\familydefault}{\mddefault}{\updefault}{\color[rgb]{0,0,0}$Y_1$}%
}}}}
\put(3151,-1036){\makebox(0,0)[lb]{\smash{{\SetFigFont{12}{14.4}{\familydefault}{\mddefault}{\updefault}{\color[rgb]{0,0,0}$Y_2$}%
}}}}
\put(3151,-2236){\makebox(0,0)[lb]{\smash{{\SetFigFont{12}{14.4}{\familydefault}{\mddefault}{\updefault}{\color[rgb]{0,0,0}$Y_K$}%
}}}}
\put(901,-1336){\makebox(0,0)[lb]{\smash{{\SetFigFont{12}{14.4}{\familydefault}{\mddefault}{\updefault}{\color[rgb]{0,0,0}$U$}%
}}}}
\end{picture}%
}
\caption{The binary-symmetric parallel-relay network with $K$ relays}
\label{fig:parallel-relay-network}
\end{figure}

Networks with relays were first introduced by van der Meulen~\cite{meulen71} in which three nodes exchange data and any node can facilitate the message transfer between the other two. The special case where a source transmits data to a destination with the help of a relay, which itself has no private data to send, was considered by Cover and El Gamal~\cite{covergamal79}, and it is now commonly known as \emph{the relay channel}. Since then, many variants of relay channels have been investigated, for example: (a) the multiple-access relay channel~\cite{sankarkramer07,hausl09}, (b) the broadcast relay channel~\cite{liangkramer07}, (c) the multiple-relay channel~\cite{xiekumar03,kramergastpar04}, (d) the multi-way relay channel~\cite{avestimehrsezgin10,ongmjohnsonit11}, (e) the multiple-input multiple-output (MIMO) relay network~\cite{bolcskeinabar06}, (f) the multiple-source multiple-destination multiple-relay wireless network~\cite{xiekumar04,liangkumar07,johnsonong11ett}, (g) the additive white Gaussian noise (AWGN) parallel-relay network~\cite{gastparvetterli05,sanderovichshamai08,scheingallager00,boujemaa10}, (h) the discrete-memoryless parallel-relay network with lossless channels from the relays to the destination~\cite{sanderovichshamai08}, and (i) the two-user BSPR network with  lossless channels from the relays to the destination~\cite{schein01thesis}.
To the best of our knowledge, the BSPR network with noise on \emph{both} the source-to-relay and relay-to-destination channels has not previously been investigated. More specifically, this work differs from (h) and (i) in that the relay-to-destination channels are noisy and that the asymptotic behavior of the network when the number of relays grows is investigated.

To date, the capacity of even the simple single-relay channel is not known except for a few special cases, e.g., the degraded case~\cite{covergamal79,xiekumar03} (note that the parallel-relay network is not a degraded channel).
This hints at the difficulty of analyzing multiterminal networks, especially large networks.

Asymptotic capacity results, however, have been obtained for some AWGN networks.
It has been shown that using coding at the source and no coding at the relays achieves the asymptotic capacity of the AWGN parallel-relay network~\cite{gastparvetterli05} and the non-coherent MIMO relay network~\cite{bolcskeinabar06} as the number of relays increases to infinity. Here the source encodes the data and the relays merely scale their received signals and forward them. These results hold only for large AWGN networks (with the number of relays tending to infinity) and have not been proven for other types of networks. In this paper, we will show a similar result for the BSPR network\footnote{The key difference (besides the channel alphabets) between the AWGN parallel-relay network and the BSPR network is that, for the latter the relays' signals are received orthogonally at the destination, see Fig.~\ref{fig:parallel-relay-network}.}, i.e., no coding (only forwarding) at the relays achieves the capacity asymptotically as the number of relays tends to infinity. In addition, we will also show that using coded transmission at the source and no coding at the relays achieves transmission rates close to the capacity of the BSPR network with a generally medium number of relays, a result which has no current equivalent in the AWGN network. For example, for a network with \emph{cross-over} probability of 0.2, we show that using 37 or more forwarding relays allows us to achieve within 0.0001 bits of the capacity.


While it has also been shown~\cite{gastpar08} that uncoded transmission is optimal (in terms of the distortion measure) in the AWGN parallel-relay network where the source messages are Gaussian random variables, we show in this paper that in the binary-symmetric parallel-relay channel where the source messages are binary, uncoded transmission is strictly suboptimal when the number of relays is small.




\subsection{Main Results}

In this paper, we show that by applying coding at the source and using a finite number of forwarding relays, i.e., relays which forward their received signals without any decoding and re-encoding, performances arbitrarily close to the capacity of the BSPR network can be achieved. Alternatively, asymptotic (as the number of relays increases) capacity results hold even when both the source and the relays
send uncoded message bits. For small networks on the other hand, forwarding relays do not achieve the capacity, but decoding relays, where the relays decode the source message, re-encode, and forward it to the destination, do achieve the capacity.
Specifically, we have the following results, where the terms small, medium, and large used above have been made precise (recall that $K$ is the number of relays in the BSPR network).
\begin{itemize}
\item Theorems~\ref{theorem:cut-set-upper-bound} and \ref{theorem:upper-bounds-capacity} give two upper bounds to the capacity of the BSPR network.
\item Theorems~\ref{thm:achievability-simple}, \ref{theorem:coded-relay}, and \ref{theorem:hybrid} give three lower bounds (achievable rates) to the capacity of the BSPR network.
\item Corollary~\ref{corollary:df} gives the capacity of small BSPR networks, i.e., when $1 \leq K \leq K'(p_s,p_d)$ for some $K'(p_s,p_d)$, defined in \eqref{eq:kdash}, which is a function of the cross-over probabilities.
\item Theorems~\ref{thm:uncoded} and \ref{theorem:coded-forward-asymptotic} give asymptotic capacity results of the BSPR network as $K$ increases.
\item Corollary~\ref{corollary:forwarding} gives achievable rates arbitrarily close to the capacity for medium to large networks, i.e., achievable rates within $\zeta$ bits of the capacity for an arbitrarily small $\zeta>0$, when $K > K(p,\zeta)$ for some $K(p,\zeta)$, defined in \eqref{eq:q}, which is a function of the cross-over probabilities and $\zeta$.
\end{itemize}


\section{Channel Model and Notation}\label{sec:channel-model}

We consider the BSPR network depicted in Fig.~\ref{fig:parallel-relay-network}. The network consists of a source (denoted by node $S$), $K$ relays (denoted by nodes $1,2\dotsc,K$), a destination (denoted by node $D$), and a network of channels defined as follows:

The channel from the source to the $i$th relay is a binary symmetric channel given by
\begin{equation}
V_i = U \oplus Z_i,
\end{equation}
where $U \in \{0,1\}$ is the signal transmitted by the source, $V_i \in \{0,1\}$ the signal received by the $i$th relay, $Z_i \in \{0,1\}$ is channel noise with $\Pr\{Z_i = 1\} = p_{s,i}$, and $\oplus$ is modulo-two addition. The probability $p_{s,i}$ is also known as the cross-over probability for the binary symmetric channel from the source to relay $i$. 
Without loss of generality, we assume that $0 \leq p_{s,i} \leq \frac{1}{2}$. 

The channel from relay $i$ to the destination is a binary symmetric channel given by
\begin{equation}
Y_i = X_i \oplus E_i,
\end{equation}
where $X_i \in \{0,1\}$ is the signal transmitted by relay $i$, $Y_i \in \{0,1\}$ is the signal received by the destination, and $E_i \in \{0,1\}$ is channel noise with $\Pr\{E_i = 1\} = p_{i,d}$. Again, we assume that $0 \leq p_{i,d} \leq \frac{1}{2}$.

We assume that all the channels are independent, time invariant, and memoryless. This means
all $Z_i$ and $E_i$ are independent of each other and are also independent and identically distributed over network uses (one network use is defined as simultaneous uses of all the channels in the network). 

For any channel variable $A \in \{U, Z_i, V_i, X_i, E_i, Y_i: 1 \leq i \leq K\}$, we denote by $A[t]$ the variable $A$ at time $t$, i.e., $A$ on the $t$th network use.
We use a bold symbol to denote a vector (over time) of length $n$, i.e., $\boldsymbol{A} \triangleq ( A[1], A[2], \dotsc, A[n] )$, where $n$ is the code length. We will also use bar to denote a vector (over the relays) of length $K$, i.e., $\bar{A} \triangleq (A_1, A_2, \dotsc, A_K)$, where $K$ is the total number of relays. For a random variable $A$ in upper case, the corresponding lower case $a$ is used to denote the realization.

A message $W$ is observed by the source and is to be communicated to the destination in $n$ network uses.
An $(M,n)$ code is defined by the following encoding and decoding functions:
(i) a message, $W \in  \{0,1,\dotsc,M-1\} \triangleq \mathcal{W}$;
(ii) an encoding function at the source, $\boldsymbol{U} = f_S(W)$;
(iii) a set of $n$ encoding functions at each relay, $X_i[t] = f_{i,t}(V_i[1], V_i[2], \dotsc, V_i[t-1])$, for all $i \in \{1,2,\dotsc,K\}$ and $t \in \{1,2,\dotsc, n\}$;
(iv) a decoding function at the destination, $\hat{W} = g(\boldsymbol{Y}_1, \boldsymbol{Y}_2, \dotsc, \boldsymbol{Y}_K)$,
where $\hat{W}$ is the estimate of $W$.
Note that the relays' transmit signals can only depend on their respective past received signals. The {\em rate} of this code is $(\log_2 M)/n$ bits per network use.

Let $W$ be randomly and uniformly chosen from the source alphabet $\{0,1,\dotsc,M-1\}$. The average error probability is given by
$P_\textnormal{e} = \frac{1}{M} \sum_{w=0}^{M-1} \Pr\{\hat{W} \neq w | W=w\}$.
The rate $R$ is \emph{achievable} if the following is true: for any $\epsilon > 0$, there exists for sufficiently large $n$ a $(2^{nR},n)$ code such that $P_\textnormal{e} \leq \epsilon$. The \emph{capacity}, $C$, is defined as the supremum of all achievable rates.
We say that any node can decode the source message \emph{reliably} if and only if the probability of erroneous decoding can be made arbitrarily small.

A special case of the BSPR network is the symmetrical BSPR defined as follows:
\begin{definition} \label{def:symmetrical}
A BSPR network is said to be symmetrical if $p_{s,i} = p_s$ and $p_{i,d} = p_d$, for all $i \in \{1,2,\dotsc,K\}$.
\end{definition}
For the symmetrical BSPR, the channel from the source to each relay is equally noisy, and each channel from the relays to the destination is equally noisy. For the symmetrical case, we consider $p_s < 1/2$ and $p_d < 1/2$, as setting any of $p_s$ or $p_d$ to $1/2$, the capacity of the network is zero as the channels are randomized such that no information can be transferred from the source to the destination.

\section{Upper Bounds to Capacity}\label{sec:upper-bound}

In this section, we derive two upper bounds to $C$.
Define $H(p) \triangleq -p \log p - (1-p) \log (1-p)$, for some $0 < p \leq \frac{1}{2}$, and $H(0) \triangleq 0$.
We first have the following:
\begin{theorem}\label{theorem:cut-set-upper-bound}
The capacity of the BSPR network is upper bounded by $C \leq R_\textnormal{ub}$, where
\begin{equation}\label{eq:cut-set-bound}
R_\textnormal{ub} \triangleq  \min \left\{ \max_{p(u)} I(U;\bar{V}), K - \sum_{i=1}^K H(p_{i,d})   \right\}.
\end{equation}
\end{theorem}

\begin{proof}
See Appendix~\ref{appendix:ub}
\end{proof}

We have the following observation:
\begin{lemma}\label{lemma:symmetric}
The channel from $U$ to $\bar{V}$ is \emph{symmetrical} in the sense of ~\cite[Thm.\ 4.5.2]{gallager68}.
\end{lemma}

\begin{proof}
See Appendix~\ref{appendix:symmetric}.
\end{proof}

\begin{remark}
Note the difference between a symmetrical \underline{channel} in Lemma~\ref{lemma:symmetric} and a symmetrical \underline{network} as defined in Definition~\ref{def:symmetrical}. For the rest of this paper, unless otherwise stated, the term symmetrical is used to refer to symmetrical BSPR networks.
\end{remark}

From Lemma~\ref{lemma:symmetric}, it follows that the maximizing input distribution for $I(U;\bar{V})$ is the uniform binary distribution, denoted by $p'(u)$~\cite[Thm.\ 4.5.2]{gallager68}.

The aforementioned bound, $R_\textnormal{ub}$, has the same form as the cut-set bound~\cite[Thm.\ 15.10.1]{coverthomas06} except that the input distribution is the product distribution rather than the joint distribution.

Noting that $I(U;\bar{V}) \leq H(U) \leq 1$, we have the following (looser) upper bound, which is simpler and is independent of $K$:
\begin{theorem}\label{theorem:upper-bounds-capacity}
The capacity of the BSPR network is upper bounded by $C \leq 1$.
\end{theorem}

\section{An Equivalent Model for Forwarding Relays}

Before we present achievable rates of the different coding schemes, we first derive an equivalent model for the BSPR network when all relays do simple forwarding. 
\begin{definition}
Relay $i$, $i \in \{1,2,\dotsc,K\}$, is called a \emph{forwarding relay} if it simply forwards its received signals, i.e.,
\begin{equation}
X_i[t] = V_i[t-1],
\end{equation}
for all $t$, and $X_i[1] = 0$ by convention.
\end{definition}

Using only forwarding relays, the received signals at the destination can be re-written as
\begin{equation}
Y_i = U \oplus Z_i \oplus E_i = U \oplus N_i,\label{eq:equivalent-channel-2}
\end{equation}
for all $ i \in \{1,2,\dotsc,K\}$,
where $N_i \triangleq Z_i \oplus E_i$.
We have dropped the time indices in the aforementioned equation as it is clear that the destination receives the noisy version of the source's transmission two units of time later. It follows that
\begin{equation}\label{eq:equivalent-cross-over}
\Pr\{N_i = 1\} = p_{s,i}(1-p_{i,d}) + (1-p_{s,i})p_{i,d} \triangleq p_i,
\end{equation}
which is the cross-over probability of the \emph{effective} binary symmetric channel from $U$ to $Y_i$ when relay $i$ is a forwarding relay.


\begin{figure}[t]
\centering
\resizebox{6cm}{!}{ 
\begin{picture}(0,0)%
\includegraphics{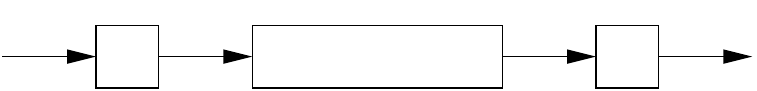}%
\end{picture}%
\setlength{\unitlength}{3947sp}%
\begingroup\makeatletter\ifx\SetFigFont\undefined%
\gdef\SetFigFont#1#2#3#4#5{%
  \reset@font\fontsize{#1}{#2pt}%
  \fontfamily{#3}\fontseries{#4}\fontshape{#5}%
  \selectfont}%
\fi\endgroup%
\begin{picture}(3624,420)(-11,-1573)
\put(1276,-1486){\makebox(0,0)[lb]{\smash{{\SetFigFont{12}{14.4}{\familydefault}{\mddefault}{\updefault}{\color[rgb]{0,0,0}$Y_i = U \oplus N_i$}%
}}}}
\put(2926,-1486){\makebox(0,0)[lb]{\smash{{\SetFigFont{12}{14.4}{\familydefault}{\mddefault}{\updefault}{\color[rgb]{0,0,0}$D$}%
}}}}
\put(3226,-1336){\makebox(0,0)[lb]{\smash{{\SetFigFont{12}{14.4}{\familydefault}{\mddefault}{\updefault}{\color[rgb]{0,0,0}$\hat{W}$}%
}}}}
\put(2551,-1336){\makebox(0,0)[lb]{\smash{{\SetFigFont{12}{14.4}{\familydefault}{\mddefault}{\updefault}{\color[rgb]{0,0,0}$\bar{Y}$}%
}}}}
\put(526,-1486){\makebox(0,0)[lb]{\smash{{\SetFigFont{12}{14.4}{\familydefault}{\mddefault}{\updefault}{\color[rgb]{0,0,0}$S$}%
}}}}
\put( 76,-1336){\makebox(0,0)[lb]{\smash{{\SetFigFont{12}{14.4}{\familydefault}{\mddefault}{\updefault}{\color[rgb]{0,0,0}$W$}%
}}}}
\put(901,-1336){\makebox(0,0)[lb]{\smash{{\SetFigFont{12}{14.4}{\familydefault}{\mddefault}{\updefault}{\color[rgb]{0,0,0}$U$}%
}}}}
\end{picture}%
}
\caption{The equivalent point-to-point channel with forwarding relays, where $\bar{Y}=(Y_1,Y_2,\dotsc,Y_K)$}
\label{fig:equivalent-channel-simple-relays}
\end{figure}
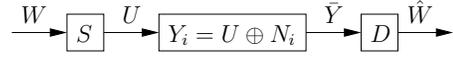

Thus, we have turned the network in Fig.~\ref{fig:parallel-relay-network} with only forwarding relays into a point-to-point channel from $U$ to $\bar{Y}$ as depicted in Fig.~\ref{fig:equivalent-channel-simple-relays}. Recall that $\bar{Y} \triangleq(Y_1,Y_2,\dotsc,Y_K)$. 

\section{Lower Bounds to Capacity}\label{sec:lower-bound}

In this section, we present four coding schemes and derive achievable rates for each scheme.

\subsection{Coded Transmission with Forwarding Relays}\label{sec:coded-transm-with-simple}

We first consider coded transmission with forwarding relays, where (i) the source performs channel coding, and (ii) all the relays do simple forwarding. Therefore, we have an equivalent point-to-point channel from $U$ to $\bar{Y}$ as described in the previous section. This coding scheme is also known as amplify-and-forward~\cite{kramergastpar04}.
Using this scheme, the following rates are achievable:
\begin{theorem}\label{thm:achievability-simple}
Consider the BSPR network. Coded transmission with forwarding relays achieves any rate $R < R_\textnormal{coded,f}$, where
\begin{equation}\label{eq:achievable-rate-coded-forward}
R_\textnormal{coded,f} \triangleq  I(U;\bar{Y})
\end{equation}
evaluated with the uniform binary distribution $p'(u)$. 
\end{theorem}

\begin{proof}
Using the results of the point-to-point channel~\cite[Thm.\ 7.7.1]{coverthomas06}, the aforementioned scheme achieves rates up to  $R < \frac{n}{n+1} I(U;\bar{Y})$, as we use $(n+1)$ network uses in total to transmit an $nR$-bit message. However, for any $R$ that satisfies $R < R_\textnormal{coded,f}$, we can always find a sufficiently large $n$ such that $R < \frac{n}{n+1}R_\textnormal{coded,f} = \frac{n}{n+1} I(U;\bar{Y})$. 
Also, using a similar argument as the proof of Lemma~\ref{lemma:symmetric}, we can show that the channel $U \rightarrow \bar{Y}$ is symmetrical. So, the maximizing input distribution for $I(U;\bar{Y})$ is the uniform distribution~\cite[Thm.\ 4.5.2]{gallager68}. %
\end{proof}

\vspace{1ex}\noindent {\bf Symmetrical case: $p_{s,i} = p_s$ and $p_{i,d} = p_d$, $\forall i$}
\newline\indent Now, we evaluate \eqref{eq:achievable-rate-coded-forward} for the symmetrical BSPR network:
\begin{corollary}\label{theorem:simple-evaluation}
Consider the symmetrical BSPR network. Coded transmission with forwarding relays achieves any rate $R < R_\textnormal{coded,f}$, where 
\begin{multline}\label{eq:simple-evaluation}
  R_\textnormal{coded,f}  = 1 + Kp \log p + Kq\log q  - \sum_{l=0}^{K-1} \Bigg[\binom{K-1}{l}  \\ (q^lp^{K-l} + q^{K-l}p^l) \log (q^lp^{K-l} + q^{K-l}p^l) \Bigg],
\end{multline} 
where $p\triangleq p_{s}(1-p_{d}) + (1-p_{s})p_{d}$ and $q \triangleq 1-p$.
\end{corollary}

\begin{proof}
Eqn.~\eqref{eq:simple-evaluation} follows by evaluating the RHS of \eqref{eq:achievable-rate-coded-forward} with $p_{s,i} = p_s$, $p_{i,d} = p_d$, $\forall i \in \{1,2,\dotsc,K\}$, and the uniform distribution $p'(u)$.
\end{proof}


\begin{remark}
Note that evaluation of the aforementioned rate has factorial time complexity in the number of relays, $K$. This also has implications for the evaluation of the capacity upper bound $R_\textnormal{ub}$ in Theorem~\ref{theorem:cut-set-upper-bound}.
\end{remark}


\vspace{1ex} \noindent {\bf General case: possibly different $p_{s,i}$ and $p_{i,d}$}
\newline\indent 
Now, we derive a lower bound for $R_\textnormal{coded,f}$ that allows us to determine the performance of this scheme in the general BSPR network as the number of relays increases. First, we define
\begin{equation}
p_{\textnormal{MAX}(\mathcal{M})} \triangleq \max_{m_i \in \mathcal{M}} p_{m_i}, \label{eq:m}
\end{equation}
for some $\mathcal{M} =\{m_1,m_2,\dotsc, m_M\} \subseteq \{1,2,\dotsc,K\}$. This means there exist $M$ (out of $K$) relays where the effective channels from the source to the destination via these $M$ relays are each not noisier than a binary-symmetric channel with cross-over probability $p_{\textnormal{MAX}(\mathcal{M})}$.
We can show that the following rate is achievable.

\begin{corollary}\label{theorem:general-coded-forwarding-rate}
Consider the BSPR network. Coded transmission with forwarding relays achieves any rate $R < R_\textnormal{coded,f}$ where 
\begin{multline}
R_\textnormal{coded,f} \geq  1 + Mp_{\textnormal{MAX}(\mathcal{M})} \log p_{\textnormal{MAX}(\mathcal{M})}\\
\quad\quad \quad + M(1-p_{\textnormal{MAX}(\mathcal{M})})\log (1-p_{\textnormal{MAX}(\mathcal{M})})\\ - \sum_{l=0}^{M-1} \Bigg[ \binom{M-1}{l} (q_{\textnormal{MAX}(\mathcal{M})}^lp_{\textnormal{MAX}(\mathcal{M})}^{M-l} + q_{\textnormal{MAX}(\mathcal{M})}^{M-l}p_{\textnormal{MAX}(\mathcal{M})}^l)\\
\quad\quad \quad\quad\log (q_{\textnormal{MAX}(\mathcal{M})}^lp_{\textnormal{MAX}(\mathcal{M})}^{M-l} + q_{\textnormal{MAX}(\mathcal{M})}^{M-l}p_{\textnormal{MAX}(\mathcal{M})}^l) \Bigg],\label{eq:coded-general}
\end{multline} 
for any $\mathcal{M} \subseteq \{1,2,\dotsc,K\}$,
where $q_{\textnormal{MAX}(\mathcal{M})} = 1-p_{\textnormal{MAX}(\mathcal{M})}$, and $M = |\mathcal{M}|$.
\end{corollary}

\begin{proof}[Proof of Corollary~\ref{theorem:general-coded-forwarding-rate}]
If a rate $R$ is achievable using some $M$ relays, then $R$ is also achievable in the original network with $K$ relays. Among these $M$ relays, the maximum $p_{m_i}$ is denoted by $p_{\textnormal{MAX}(\mathcal{M})}$. For each sub-channel from $U$ to $Y_{m_i}$ with cross-over probability $p_{m_i}$, we further add random noise to get $Y'_{m_i} = Y_{m_i} \oplus E_{m_i}''$, with some independent and random $E_{m_i}'' \in \{0,1\}$ where $\Pr\{E_{m_i}''=1\} = \frac{p_{\textnormal{MAX}(\mathcal{M})}-p_{m_i}}{1-2p_{m_i}}$.
In this modified network, we have
\begin{equation}
Y'_{m_i} = U \oplus N'_{m_i},
\end{equation}
where $N'_{m_1} = Z_{m_i} \oplus E_{m_i} \oplus E_{m_i}''$ and $\Pr\{N'_{m_i} =1\} =p_{\textnormal{MAX}(\mathcal{M})}$, $\forall m_i \in \mathcal{M}$. If the rate $R'$ is achievable on the modified network from $U$ to $(Y'_{m_1}. Y'_{m_2}, \dotsc, Y'_{m_M})$, then the rate $R \geq R'$ is achievable in the original network using only $M$ relays and is also achievable in the original network with $K$ relays. Using the results from Corollary~\ref{theorem:simple-evaluation}, we have Corollary~\ref{theorem:general-coded-forwarding-rate}.
\end{proof}

\begin{remark} \label{remark:m}
We will show in the next section (see Remark~\ref{remark:m-proof}) that as long as we can find a sequence of $\mathcal{M} \subseteq \{1,2,\dotsc,K\}$ (one for each $K$), such that $|\mathcal{M}|(0.5-p_{\textnormal{MAX}(\mathcal{M})}) \rightarrow \infty$ as $K \rightarrow \infty$, then coded transmission with forwarding relays approaches the capacity.
\end{remark}

\begin{remark} \label{remark:m-2}
The condition in Remark~\ref{remark:m} is not unreasonable as long as when the number of relays increases, the number of ``bad'' channels (with cross-over probabilities $p_i$ close to 0.5) can be kept at or below a certain fraction of $K$. Another example of a network satisfying the condition is one in which as $K$ increases, the maximum cross-over probability is unchanged at some $p_{\textnormal{MAX}(\{1,2,\dotsc,K\})} = p_\textnormal{const} < 0.5$.
\end{remark}


\subsection{Uncoded Transmission with Forwarding Relays}\label{sec:uncoded-transmission}

We now investigate the second scheme where the source and all the relays send uncoded signals. We first split the source message into $n$ bits, i.e., $W = (W_1,W_2,\dotsc, W_n)$, where each $W_i$ is uniformly distributed in $\{0,1\}$.  We use the following encoding functions:
\begin{enumerate}
\item Uncoded transmission at the source: $U[t] = W_t$,
\item Forwarding at the relays: $X_i[t+1] = V_i[t]$,
\end{enumerate}
for all $t \in \{1,2,\dotsc, n\}$.

\begin{remark}
In the aforementioned scheme, we send $n$ bits in $(n+1)$ network uses. So the rate here is $\frac{n}{n+1}$, which can be made arbitrarily close to 1 with a sufficiently large $n$.
\end{remark}

\begin{remark}
Note that the analysis in this section is slightly different from that in the previous section. In the previous section, we derived achievable rates as functions of network parameters. In this section, we fix the transmission rate at (arbitrarily close to) 1~bit/network use, which is a capacity upper bound, and analyze under what condition this rate is achievable.
\end{remark}

\vspace{1ex} \noindent {\bf Symmetrical case: $p_{s,i} = p_s$ and $p_{i,d} = p_d$, $\forall i$}

We have the following asymptotic capacity result as the network size increases using this scheme.

\begin{theorem}\label{thm:uncoded}
Consider the symmetrical BSPR network.  
For any $\epsilon > 0$ and at any rate $R <1 \triangleq R_\textnormal{uncoded,f}$, uncoded transmission with forwarding relays achieves $P_\textnormal{e} \leq \epsilon$ with sufficiently large $n$ and $K$.
\end{theorem}

\begin{proof}
See Appendix~\ref{appendix:uncoded-symmetrical}.
\end{proof}

From Theorem~\ref{theorem:upper-bounds-capacity}, we know that 1~bit/network use is an upper bound to the capacity. Hence uncoded transmission with forwarding relays achieves the asymptotic capacity as $K$ tends to infinity.

As the message is uncoded, the achievable rate is obtained not by typical-set decoding, but by maximum likelihood decoding where the destination simply decodes $\hat{W}=1$ if there are more 1's than there are 0's in its received signals $\bar{Y}$.

\vspace{1ex}\noindent {\bf General case: possibly different $p_{s,i}$ and $p_{i,d}$}
\newline\indent Now, consider the general case, and $p_\textnormal{MAX}(\mathcal{M})$ defined in \eqref{eq:m} for some $\mathcal{M} \subseteq \{1,2,\dotsc,K\}$. We have the following corollary:

\begin{corollary}\label{theorem:uncoded-general}
Consider the BSPR network. If there exists an $\mathcal{M} \subseteq \{1,2,\dotsc,K\}$ for each $K$ such that
\begin{equation}
|\mathcal{M}|(0.5 - p_{\textnormal{MAX}(\mathcal{M})})^2 \rightarrow \infty \quad \textnormal{ as } \quad K \rightarrow \infty, \label{eq:uncoded-general-condition}
\end{equation}
then for any $\epsilon > 0$ and at any rate $R < 1$, uncoded transmission with forwarding relays achieves $P_\textnormal{e} \leq \epsilon$ with sufficiently large $n$ and $K$.
\end{corollary}

\begin{proof}
See Appendix~\ref{appendix:uncoded-general}.
\end{proof}

The observation in Remark~\ref{remark:m-2} is also applicable to condition \eqref{eq:uncoded-general-condition}.

The maximum likelihood decoding rule for the symmetrical case (i.e., by counting the number of 1's in $\bar{Y}$) might not be optimal for the general BSPR network. However, this suboptimal decoding rule is sufficient to show the asymptotic result in Corollary~\ref{theorem:uncoded-general}.

The aforementioned results imply that with a sufficient number of relays, we do not require a channel code to achieve the capacity. This is because the relays provide sufficient spatial diversity for the destination to reliably determine the source message sent.

\subsection{Coded Transmission with Decoding Relays}\label{sec:coded-transm-with-reenc}

In the sub-sections above, we have derived achievable rates for the BSPR network with forwarding relays, and asymptotic capacity results have been obtained. In this section, we investigate a coding scheme using \emph{decoding relays}, where the relays decode the source messages, re-encode, and forward them to the destination. This coding scheme is also known as decode-and-forward~\cite{kramergastpar04}.

We obtain the following rate with decoding at some relays.
\begin{theorem}\label{theorem:coded-relay}
Consider the BSPR network. Coded transmission with decoding relays achieves any rate $R < R_\textnormal{coded,d}$, where
\begin{align}
R_\textnormal{coded,d} \triangleq \min \Bigg\{ &\Big\{ 1 - H(p_{s,m_i}): \forall m_i \in \mathcal{M} \Big\}, \nonumber \\
&M - \sum_{i=1}^M H(p_{m_i,d})\Bigg\}, \label{eq:coded-relay-rate}
\end{align}
for any $\mathcal{M} = \{m_1, m_2,\dotsc,m_M\} \subseteq \{1,2,\dotsc,K\}$ where $M= |\mathcal{M}|$. 
\end{theorem}

\begin{proof}
See Appendix~\ref{appendix:coded}
\end{proof}

\begin{remark}
The term $M - \sum_{i=1}^M H(p_{m_i,d}) = \sum_{i=1}^M \max_{p(x_{m_i})}I(X_{m_i};Y_{m_i})$ is the sum capacity of the individual point-to-point channels $X_{m_i} \rightarrow Y_{m_i}$. The term $1 - H(p_{s,m_i}) = \max_{p(u)} I(U;V_{m_i})$ is the capacity from the source to relay $m_i$.
\end{remark}

In the aforementioned coding scheme, only $M$ out of the $K$ relays are utilized, and the rest are not used. Although using more decoding relays increases the rate at which the destination can decode the source message [this can be seen from the term $M - \sum_{i=1}^M H(p_{m_i,d}) $ in \eqref{eq:coded-relay-rate}], at the same time, it imposes an additional constraint  $1 - H(p_{s,m_i})$ in \eqref{eq:coded-relay-rate} for each relay $m_i$ used, as each decoding relay must fully decode the source message. 

Although forwarding relays achieve the capacity asymptotically when $K \rightarrow \infty$, decoding relays achieve the capacity under the following conditions:

\begin{theorem}\label{theorem:finite-k-capacity}
Consider the BSPR network. Coded transmission with decoding relays achieves the capacity under the following conditions:
\begin{enumerate}
\item If $K=1$, then
\begin{equation}
C = 1 - \max \{ H(p_{s,1}), H(p_{1,d}) \} 
\end{equation}
\item If $K>1$, and if $K - \sum_{j=1}^K H(p_{j,d}) \leq 1 - H(p_{s,i}) $ for all  $i \in \{1,2,\dotsc,K\},$
then
\begin{equation}
C = K - \sum_{i=1}^K H(p_{i,d}).\label{decoding-relay-capacity}
\end{equation}
\end{enumerate}
\end{theorem}

\begin{proof}
See Appendix~\ref{appendix:finite-k-capacity}.
\end{proof}

\begin{remark}
The condition $K>1$ in the aforementioned theorem corresponds to the case where that the sum of capacities of all $K$ channels from the relays to the destination is smaller than the capacity of the channel from the source to each relay. This condition is likely to hold for small $K$ and when the channels from the relays to the destination are noisy. This result resembles that of the single-relay channel where the decode-and-forward coding scheme outperforms other schemes when the source-to-relay channel is better than the relay-to-destination channel~\cite{kramergastpar04}.  
\end{remark}




\vspace{1ex}\noindent {\bf Symmetrical case: $p_{s,i} = p_s$ and $p_{i,d}=p_d$, $\forall i$}
\newline\indent For the symmetrical case, we have the following:
\begin{corollary}\label{lemma:decoding-relays-symmetrical}
Consider the symmetrical BSPR network. Coded transmission with decoding relays achieves any rate $R < R_\textnormal{coded,d}$, where
\begin{equation}
R_\textnormal{coded,d} = \min \{ 1 - H(p_s) , K ( 1 - H(p_d)) \}. \label{eq:df-symmetry}
\end{equation}
\end{corollary}

\begin{proof}
When $p_{s,i}=p_s$, $\forall i$, we have $1 - H(p_{s,i}) = 1 - H(p_s),\forall i$.
So, selecting any $\mathcal{M}$ will not affect $\min_{m_i \in \mathcal{M}} \{ 1 - H(p_{s,m_i})\}$ in \eqref{eq:coded-relay-rate}. 
In addition, adding more relays into the set $\mathcal{M}$ increases $[M-\sum_{i=1}^MH(p_{m_i,d})]$ as $H(p_{m_i,d}) < 1$. Hence, it is always optimal to set $\mathcal{M} = \{1,2,\dotsc,K\}$. So, \eqref{eq:df-symmetry} follows from \eqref{eq:coded-relay-rate} with $M=K$.
\end{proof}

Furthermore, when $p_d \leq p_s$, we have $1-H(p_d) \geq 1-H(p_s)$, and the achievable rate further simplifies to
\begin{equation}
R_\textnormal{coded,d} = 1 - H(p_s), \label{eq:ps-equalspd}
\end{equation}
which is independent of $K$.

\subsection{A Hybrid Coding Scheme}


Using coded transmission and forwarding relays, noise on the source-to-relays channels propagates to the relays-to-destination channels. This can be rectified by having the relays decode the source message before forwarding it (hence removing the noise on the source-to-relay channels). Doing so, however, imposes additional rate constraints that the relays must fully decode the source message. We next propose a hybrid coding scheme in  which some relays  decode and re-encode the messages, and the rest of the relays forward their received signals.

\subsubsection{Achievable Rates}

Denote the set of decoding relays by $\mathcal{M} \subseteq \{1,2,\dotsc, K\}$ and the set of forwarding relays by $\mathcal{F} = \{1,2,\dotsc,K\} \setminus \mathcal{M}$. Since all relays in the set $\mathcal{M}$ need to decode the messages, the rate is constrained by $ 1 - H\left( \max_{i \in \mathcal{M}} p_{s,i}\right)$.
Each relay $i \in \mathcal{M}$ fully decodes the source message and transmits using codewords $\boldsymbol{X}_i$ that are generated independent of $U$ and of the codewords of other relays $j \in \mathcal{M}$. Under this scheme, we can view the BSPR  network as a point-to-point MIMO channel from $(U,X_\mathcal{M})$ to $(Y_1,Y_2,\dotsc,Y_K)$, where $U$ and $X_i$, for all $i \in \mathcal{M}$, are statistically independent. This MIMO channel consists of  sub-channels $X_i \rightarrow Y_i$ for all $i \in \mathcal{M}$ (from the decoding relays to the destination), as well as $U \rightarrow X_\mathcal{F} \rightarrow Y_\mathcal{F}$ (from the source to the destination through forwarding relays). So, the destination can decode the source message at the rate $ I(U,X_{\mathcal{M}};Y_1,Y_2,\dotsc,Y_K) = I(U;Y_{\mathcal{F}}) + I(X_{\mathcal{M}};Y_{\mathcal{M}}) = I(U;Y_{\mathcal{F}}) + \sum_{i \in \mathcal{M}} I(X_i,Y_i)$. Hence, we have the following theorem.

\begin{theorem}\label{theorem:hybrid}
Consider the BSPR network. The hybrid scheme achieves any rate $R < R_\textnormal{coded,h}$, where
\begin{align}
R_\textnormal{coded,h} \triangleq \min &\Bigg\{ 1 - H\left( \max\limits_{i \in \mathcal{M}} p_{s,i}\right), \nonumber\\ &I(U;Y_{\mathcal{F}}) + \sum\limits_{i \in \mathcal{M}} \Big(1 - H(p_{i,d})\Big) \Bigg\}, \label{eq:hybrid}
\end{align}
for any $\mathcal{M} \subseteq \{1,2,\dotsc,K\}$ and  $\mathcal{F} = \{1,2,\dotsc,K\} \setminus \mathcal{M}$,
where $H\left( \max_{i \in \mathcal{M}} p_{s,i}\right)=0$ if $\mathcal{M} = \varnothing$. 
\end{theorem}

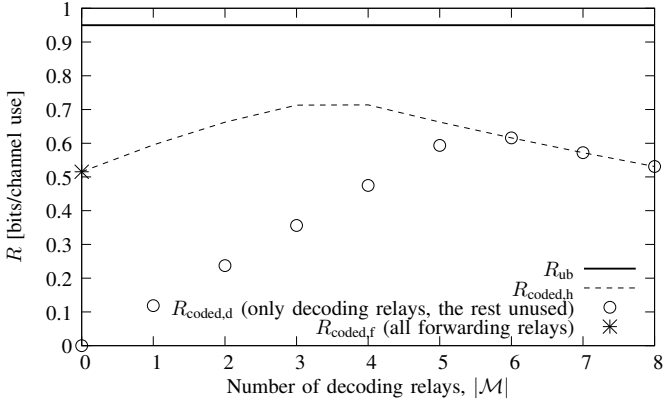
\begin{figure}[t]
\centering
\resizebox{\linewidth}{!}{ 
\begin{picture}(0,0)%
\includegraphics{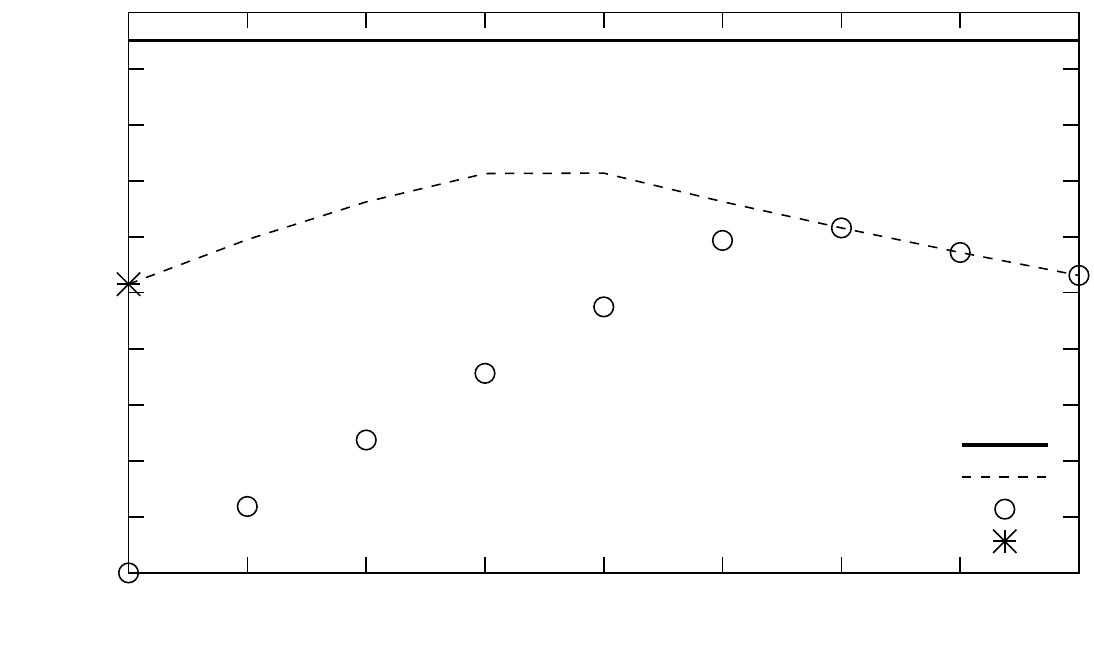}%
\end{picture}%
\setlength{\unitlength}{3947sp}%
\begingroup\makeatletter\ifx\SetFigFont\undefined%
\gdef\SetFigFont#1#2#3#4#5{%
  \reset@font\fontsize{#1}{#2pt}%
  \fontfamily{#3}\fontseries{#4}\fontshape{#5}%
  \selectfont}%
\fi\endgroup%
\begin{picture}(5247,3181)(1196,-3995)
\put(4094,-3935){\makebox(0,0)[b]{\smash{{\SetFigFont{10}{12.0}{\familydefault}{\mddefault}{updefault}Number of decoding relays, $|\mathcal{M}|$}}}}
\put(1738,-3623){\makebox(0,0)[rb]{\smash{{\SetFigFont{10}{12.0}{\familydefault}{\mddefault}{\updefault} 0}}}}
\put(1738,-3354){\makebox(0,0)[rb]{\smash{{\SetFigFont{10}{12.0}{\familydefault}{\mddefault}{\updefault} 0.1}}}}
\put(5738,-3471){\makebox(0,0)[rb]{\smash{{\SetFigFont{10}{12.0}{\familydefault}{\mddefault}{\updefault}$R_\text{coded,f}$ (all forwarding relays)}}}}
\put(1738,-3085){\makebox(0,0)[rb]{\smash{{\SetFigFont{10}{12.0}{\familydefault}{\mddefault}{\updefault} 0.2}}}}
\put(1738,-2816){\makebox(0,0)[rb]{\smash{{\SetFigFont{10}{12.0}{\familydefault}{\mddefault}{\updefault} 0.3}}}}
\put(1738,-2547){\makebox(0,0)[rb]{\smash{{\SetFigFont{10}{12.0}{\familydefault}{\mddefault}{\updefault} 0.4}}}}
\put(1738,-2278){\makebox(0,0)[rb]{\smash{{\SetFigFont{10}{12.0}{\familydefault}{\mddefault}{\updefault} 0.5}}}}
\put(1738,-2010){\makebox(0,0)[rb]{\smash{{\SetFigFont{10}{12.0}{\familydefault}{\mddefault}{\updefault} 0.6}}}}
\put(1738,-1741){\makebox(0,0)[rb]{\smash{{\SetFigFont{10}{12.0}{\familydefault}{\mddefault}{\updefault} 0.7}}}}
\put(1738,-1472){\makebox(0,0)[rb]{\smash{{\SetFigFont{10}{12.0}{\familydefault}{\mddefault}{\updefault} 0.8}}}}
\put(1738,-1203){\makebox(0,0)[rb]{\smash{{\SetFigFont{10}{12.0}{\familydefault}{\mddefault}{\updefault} 0.9}}}}
\put(1738,-934){\makebox(0,0)[rb]{\smash{{\SetFigFont{10}{12.0}{\familydefault}{\mddefault}{\updefault} 1}}}}
\put(1813,-3748){\makebox(0,0)[b]{\smash{{\SetFigFont{10}{12.0}{\familydefault}{\mddefault}{\updefault} 0}}}}
\put(2383,-3748){\makebox(0,0)[b]{\smash{{\SetFigFont{10}{12.0}{\familydefault}{\mddefault}{\updefault} 1}}}}
\put(5738,-3009){\makebox(0,0)[rb]{\smash{{\SetFigFont{10}{12.0}{\familydefault}{\mddefault}{\updefault}$R_\text{ub}$}}}}
\put(2954,-3748){\makebox(0,0)[b]{\smash{{\SetFigFont{10}{12.0}{\familydefault}{\mddefault}{\updefault} 2}}}}
\put(3524,-3748){\makebox(0,0)[b]{\smash{{\SetFigFont{10}{12.0}{\familydefault}{\mddefault}{\updefault} 3}}}}
\put(4094,-3748){\makebox(0,0)[b]{\smash{{\SetFigFont{10}{12.0}{\familydefault}{\mddefault}{\updefault} 4}}}}
\put(5738,-3163){\makebox(0,0)[rb]{\smash{{\SetFigFont{10}{12.0}{\familydefault}{\mddefault}{\updefault}$R_\text{coded,h}$}}}}
\put(5738,-3317){\makebox(0,0)[rb]{\smash{{\SetFigFont{10}{12.0}{\familydefault}{\mddefault}{\updefault}$R_\text{coded,d}$ (only decoding relays, the rest unused)}}}}
\put(4664,-3748){\makebox(0,0)[b]{\smash{{\SetFigFont{10}{12.0}{\familydefault}{\mddefault}{\updefault} 5}}}}
\put(5235,-3748){\makebox(0,0)[b]{\smash{{\SetFigFont{10}{12.0}{\familydefault}{\mddefault}{\updefault} 6}}}}
\put(5805,-3748){\makebox(0,0)[b]{\smash{{\SetFigFont{10}{12.0}{\familydefault}{\mddefault}{\updefault} 7}}}}
\put(6375,-3748){\makebox(0,0)[b]{\smash{{\SetFigFont{10}{12.0}{\familydefault}{\mddefault}{\updefault} 8}}}}
\put(1331,-2217){\rotatebox{90.0}{\makebox(0,0)[b]{\smash{{\SetFigFont{10}{12.0}{\familydefault}{\mddefault}{\updefault}$R$ [bits/channel use]}}}}}
\end{picture}%
}
\caption{Comparing the hybrid coding scheme to the other coding schemes and the cut-set upper bound, $K=8$, $p_{s,i}= \frac{0.1i}{8}$, $p_{i,d}=0.3$, for $i \in \{1,2,\dotsc,8\}$}
\label{fig:hybrid}
\end{figure}

\begin{remark}
Setting all relays as forwarding relays, i.e., $\mathcal{M}=\varnothing$ and $\mathcal{F} = \{1,2,\dotsc,K\}$, we recover Theorem~\ref{thm:achievability-simple}. Setting only relays in $\mathcal{M} \subseteq \{1,2,\dotsc,K\}$ as decoding relays and ignoring the other relays (i.e., not using them), i.e., $\mathcal{F} = \varnothing$, we recover Theorem~\ref{theorem:coded-relay}.
\end{remark}


Determining the optimal set of $\mathcal{M}$ is a hard combinatorial problem. A rule of thumb is to include relays with low $p_{s,i}$ in $\mathcal{M}$, and relays with high $p_{s,i}$ in $\mathcal{F}$. The reason is that a high $p_{s,i}$ (a noisy channel from the source to relay $i$) constrains the overall rate when $i$  is included in $\mathcal{M}$ (i.e., setting relay $i$ to decode).

\subsubsection{Numerical Example}
We now show that the hybrid scheme is useful when the channels from the source to the relays have different noise levels. As an example, we consider the eight-relay BSPR network with $p_{s,i}=0.1i/8$, and $p_{i,d} = 0.3$, for $i \in \{1,2,\dotsc,8\}$. We compare $R_\textnormal{coded,f}$, $R_\textnormal{coded,d}$, and $R_\textnormal{coded,h}$ using different numbers of decoding relays. The results are depicted in Fig.~\ref{fig:hybrid}. Setting all relays to perform forwarding, we achieve $R_\textnormal{coded,f} = 0.52$ (i.e., the hybrid scheme when $|\mathcal{M}|=0$). With only decoding relays, the rate is maximized at $R_\textnormal{coded,d} = 0.62$ by using six relays, i.e., when six of the eight relays perform decoding and re-encoding and the rest of the relays are unused (note that using relay 8 as a decoding relay will constrain the rate to be $R_\textnormal{coded,d} \leq 1 - H(0.1) = 0.531$). Using a combination of forwarding and decoding relays, we can achieve a significantly higher rate of $R_\textnormal{coded,h} = 0.71$ by using four decoding relays and four forwarding relays. 

\subsubsection{Other Combinations}

In general, other rates can also be achieved using different coding schemes in which the relays perform other functions. Some examples are as follows:
\begin{itemize}
\item The relays decode part (possibly different parts for different relays) of the source messages, re-encode, and forward them to the destination (demonstrated on AWGN parallel-relay networks in~\cite{ghabeliaref08}).
\item The relays compress the received signals, and forward them to the destination (demonstrated on AWGN parallel-relay networks in~\cite{kochmankhinaerezzamir08}).
\item A combination of different schemes at different relays (demonstrated on AWGN parallel-relay networks in~\cite{rezaeigharan09}).
\end{itemize}

\subsubsection{An upper bound to achievable rates with decoding relays}
For any coding scheme, if any of the relays is required to reliably decode the source messages,\footnote{This condition is not necessary for reliable communication because reliable decoding at the destination does not necessarily require reliable decoding at any relay.} the achievable rate is necessarily constrained by the following:

\begin{theorem}\label{lemma:upper-bound-full-decoding}
For any $(2^{nR},n)$ code, if any relay $i \in \{1,\dotsc,K\}$ is to reliably decode the source message, then the rates achievable by this code are upper bounded as
\begin{subequations}
\begin{align}
R &\leq \max_{p(u)} I(U;V_i)\label{eq:upper-bound-on-relay}\\
&= 1 - H(p_{s,i}).\label{eq:upper-bound-on-relay-binary-symmetric}
\end{align}
\end{subequations}
\end{theorem}

\begin{proof}
Since we need to ensure that the probability of relay $i$ wrongly decoding the source message can be made arbitrarily small, the rate $R$ must be upper bounded by the channel capacity from the source to relay $i$. From~\cite[Thm.\ 7.7.1]{coverthomas06}, we have \eqref{eq:upper-bound-on-relay}. Since the channel from $U$ to $V_i$ is a binary symmetric channel, we have \eqref{eq:upper-bound-on-relay-binary-symmetric}.
\end{proof}


\section{Capacity vs. Number of Relays}

In this section, we consider the symmetrical BSPR network. By varying the number of relays, $K$, we analyze the performance of the three schemes considered in this paper: (i) coded transmission with forwarding relays, (ii) uncoded transmission with forwarding relays, and (iii) coded transmission with decoding relays. Note that the hybrid scheme is not used here as we will either use all relays as decoding relays or as forwarding relays (see Remark~\ref{remark:no-hybrid} below).  Recall that $p \triangleq  p_{s}(1-p_{d}) + (1-p_{s})p_{d}$. 

\begin{remark} \label{remark:no-hybrid}
Note that the hybrid scheme consisting of some combination of decoding and forwarding relays is not useful in the symmetrical BSPR network. From Theorem~\ref{theorem:hybrid}, $R_\textnormal{coded,h} = \min \left\{ 1 - H(p_s), I(U;Y_{\mathcal{F}}) + |\mathcal{M}| \Big(1 - H(p_d)\Big) \right\}$ if $\mathcal{M} \neq \varnothing$, and $R_\textnormal{coded,h} = I(U;Y_{\{1,2,\dotsc,K\}})$ if $\mathcal{M}=\varnothing$, where $\mathcal{F}$ is the set of forwarding relays and $\mathcal{M}$ is the set of decoding relays. From the data-processing inequality, we know that $I(U;Y_\mathcal{F}) \leq I(X_\mathcal{F};Y_\mathcal{F}) =|\mathcal{F}|(1-H(p_d))$. So, if $\mathcal{M} \neq \varnothing$, we will choose $\mathcal{M}=\{1,2,\dotsc,K\}$. This means we either use all relays as forwarding relays or as decoding relays.
\end{remark}

\subsection{Coded Transmission with Forwarding Relays}

We first show the following asymptotic result as $K$ tends to infinity using coded transmission with forwarding relays.
\begin{theorem} \label{theorem:coded-forward-asymptotic}
Consider the symmetrical BSPR network. Coded transmissions with forwarding relays achieves any rate $R < R_\textnormal{coded,f}$ as defined in \eqref{eq:achievable-rate-coded-forward}, where
\begin{equation}
R_\textnormal{coded,f} \rightarrow C,\quad \textnormal{ as } \quad K \rightarrow \infty.
\end{equation}
\end{theorem}

\begin{proof}
See Appendix~\ref{appendix:coded-forward-asymptotic}.
\end{proof}

\begin{remark} \label{remark:m-proof}
We can use the aforementioned result to prove our claim in Remark~\ref{remark:m}. Even for a general BSPR network (which can be non-symmetrical), coded transmission with forwarding relays approaches the capacity asymptotically as the network size increases if the condition specified in Remark~\ref{remark:m} is satisfied, i.e., $|\mathcal{M}|(0.5-p_{\textnormal{MAX}(\mathcal{M})}) \rightarrow \infty$ as $K \rightarrow \infty$, where $\mathcal{M} \subseteq \{1,2,\dotsc,K\}$. From \eqref{eq:coded-general}, we note that $R_\textnormal{coded,f} \geq R'$ where $R'$ equals the rate achievable by the same coding strategy in the symmetrical BSPR network with $K$ replaced by $|\mathcal{M}|$, and $p$ by $p_{\textnormal{MAX}(\mathcal{M})}$ (c.f. \eqref{eq:simple-evaluation}). From \eqref{eq:coded-asymptotic}, we know that $R' \rightarrow C$  as $K \rightarrow \infty$. It follows that $R_\textnormal{coded,f} \rightarrow C$ as $K \rightarrow \infty$ for the general BSPR network.
\end{remark}

\begin{figure}[t]
\centering
\resizebox{\linewidth}{!}{ 
\begin{picture}(0,0)%
\includegraphics{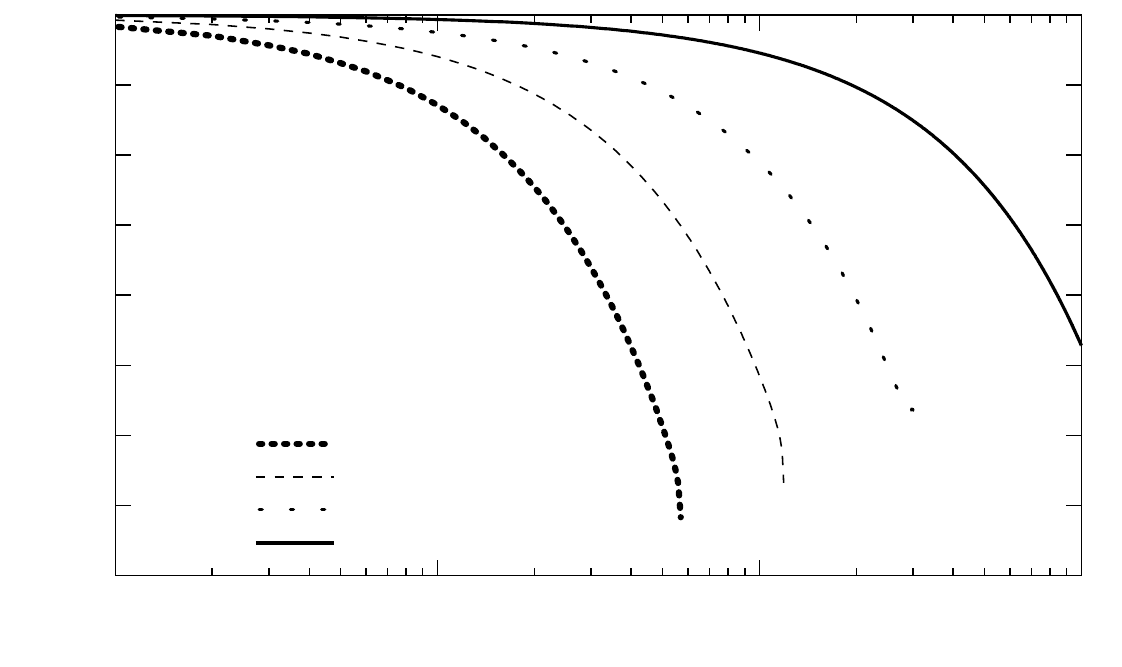}%
\end{picture}%
\setlength{\unitlength}{3947sp}%
\begingroup\makeatletter\ifx\SetFigFont\undefined%
\gdef\SetFigFont#1#2#3#4#5{%
  \reset@font\fontsize{#1}{#2pt}%
  \fontfamily{#3}\fontseries{#4}\fontshape{#5}%
  \selectfont}%
\fi\endgroup%
\begin{picture}(5375,3185)(1184,-3999)
\put(1663,-1606){\makebox(0,0)[rb]{\smash{{\SetFigFont{10}{12.0}{\familydefault}{\mddefault}{\updefault}-4}}}}
\put(1663,-1270){\makebox(0,0)[rb]{\smash{{\SetFigFont{10}{12.0}{\familydefault}{\mddefault}{\updefault}-2}}}}
\put(1663,-934){\makebox(0,0)[rb]{\smash{{\SetFigFont{10}{12.0}{\familydefault}{\mddefault}{\updefault} 0}}}}
\put(2338,-3150){\makebox(0,0)[rb]{\smash{{\SetFigFont{10}{12.0}{\familydefault}{\mddefault}{\updefault}$p=0.2$}}}}
\put(1738,-3748){\makebox(0,0)[b]{\smash{{\SetFigFont{10}{12.0}{\familydefault}{\mddefault}{\updefault} 1}}}}
\put(3284,-3748){\makebox(0,0)[b]{\smash{{\SetFigFont{10}{12.0}{\familydefault}{\mddefault}{\updefault} 10}}}}
\put(2338,-3309){\makebox(0,0)[rb]{\smash{{\SetFigFont{10}{12.0}{\familydefault}{\mddefault}{\updefault}$p=0.3$}}}}
\put(4829,-3748){\makebox(0,0)[b]{\smash{{\SetFigFont{10}{12.0}{\familydefault}{\mddefault}{\updefault} 100}}}}
\put(6375,-3748){\makebox(0,0)[b]{\smash{{\SetFigFont{10}{12.0}{\familydefault}{\mddefault}{\updefault} 1000}}}}
\put(1331,-2217){\rotatebox{90.0}{\makebox(0,0)[b]{\smash{{\SetFigFont{10}{12.0}{\familydefault}{\mddefault}{\updefault}$\log_{10}(\zeta)$}}}}}
\put(2338,-3468){\makebox(0,0)[rb]{\smash{{\SetFigFont{10}{12.0}{\familydefault}{\mddefault}{\updefault}$p=0.4$}}}}
\put(4056,-3935){\makebox(0,0)[b]{\smash{{\SetFigFont{10}{12.0}{\familydefault}{\mddefault}{\updefault}$K(p,\zeta)$}}}}
\put(1663,-3623){\makebox(0,0)[rb]{\smash{{\SetFigFont{10}{12.0}{\familydefault}{\mddefault}{\updefault}-16}}}}
\put(1663,-3287){\makebox(0,0)[rb]{\smash{{\SetFigFont{10}{12.0}{\familydefault}{\mddefault}{\updefault}-14}}}}
\put(1663,-2951){\makebox(0,0)[rb]{\smash{{\SetFigFont{10}{12.0}{\familydefault}{\mddefault}{\updefault}-12}}}}
\put(1663,-2615){\makebox(0,0)[rb]{\smash{{\SetFigFont{10}{12.0}{\familydefault}{\mddefault}{\updefault}-10}}}}
\put(1663,-2278){\makebox(0,0)[rb]{\smash{{\SetFigFont{10}{12.0}{\familydefault}{\mddefault}{\updefault}-8}}}}
\put(1663,-1942){\makebox(0,0)[rb]{\smash{{\SetFigFont{10}{12.0}{\familydefault}{\mddefault}{\updefault}-6}}}}
\put(2338,-2991){\makebox(0,0)[rb]{\smash{{\SetFigFont{10}{12.0}{\familydefault}{\mddefault}{\updefault}$p=0.1$}}}}
\end{picture}%
}
\caption{Gap from the 1-bit capacity upper bound, $\zeta$, vs. $K(p,\zeta)$ for different $p$}
\label{fig:rate-vs-relay}
\end{figure}
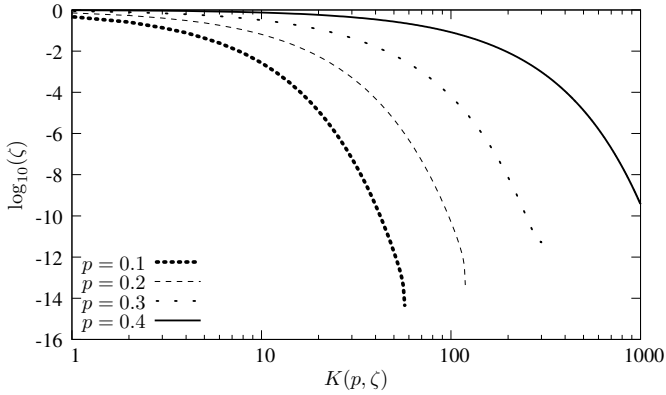

Theorem~\ref{theorem:coded-forward-asymptotic} implies that the capacity rounded to some number of significant digits is achievable with coded transmission and  $K(p,\zeta)$ or more forwarding relays, for some positive integer $K(p,\zeta)$ given in the following corollary.
\begin{corollary} \label{corollary:forwarding}
Consider the symmetrical BSPR network. For any $\zeta > 0$, if $K>K(p,\zeta)$, where $K(p,\zeta)$ is the smallest integer satisfying
\begin{multline}
1 + K(p,\zeta) p \log p + K(p,\zeta) q\log q\\ - \sum_{l=0}^{K(p,\zeta)-1} \Bigg[ \binom{K(p,\zeta)-1}{l} (q^lp^{K(p,\zeta)-l} + q^{K(p,\zeta)-l}p^l) \\ \log (q^lp^{K(p,\zeta)-l} + q^{K(p,\zeta)-l}p^l) \Bigg] \geq 1 - \zeta, \label{eq:q}
\end{multline}
then coded transmission with forwarding relays achieves rates within $\zeta$ bits of the capacity.
\end{corollary}

\begin{proof}
Because $I(U;Y_1,Y_2,\dotsc,Y_{K+1}) \geq I(U;Y_1,Y_2,$ $\dotsc,Y_K)$, $R_\textnormal{coded,f}$ is a non-decreasing function of $K$. Note that the LHS of \eqref{eq:q} is $R_\textnormal{coded,f}$ (see Corollary~\ref{theorem:simple-evaluation}). From \eqref{eq:proof-asymptotic} we know that there must exist a positive integer $K(p,\zeta)$ such that \eqref{eq:q} is true. Since $C \leq 1$ from Theorem~\ref{thm:achievability-simple}, we have $R_\textnormal{coded,f} \geq C -\zeta$ for $K=K(p,\zeta)$ as well as for all $K > K(p,\zeta)$.
\end{proof}

As $K(p,\zeta)$ in Corollary~\ref{corollary:forwarding} is not available in closed form, we numerically evaluate $K(p,\zeta)$ for $p \in \{0.1, 0.2, 0.3, 0.4\}$ and for varying $\zeta$. The results are shown in Fig.~\ref{fig:rate-vs-relay}. At $p=0.1$, we only need 16 relays to achieve within 0.0001 bits of the capacity. At $p=0.4$, we need 387 relays to achieve within 0.0001 bits of the capacity.

The proof of $R_\textnormal{coded,f}$ relies on the channel coding theorem which requires an infinitely long codelength $n$. This means we have an infinitely long \emph{delay} between the time a message is transmitted at the source and the time the message is decoded at the destination.



\subsection{Uncoded Transmission with Forwarding Relays}

In Theorem~\ref{thm:uncoded}, we have shown that when transmitting at 1~bit/network use using uncoded transmission with forwarding relays, an infinitely large number of relays are required to achieve an arbitrarily small error probability. However, if a larger error probability is acceptable, the number of relays required is smaller.
The proof of Theorem~\ref{thm:uncoded} provides an inequality explicitly upper bounding the error probability by a function of the number of relays and the crossover probability [see \eqref{eq:uncoded-error}  in Appendix~\ref{appendix:uncoded-symmetrical}].
 Table~\ref{tab:uncoded-number-of-required-relays} shows the number of relays that are sufficient to achieve $P_\textnormal{e}^\textnormal{up}$ for different channel cross-over probabilities $p \in \{0.1,0.2,0.3,0.4\}$, where $P_\textnormal{e}^\textnormal{up}$ is an upper bound to the probability of decoding error at the destination, $P_\textnormal{e}$.


\begin{table}
\renewcommand{\arraystretch}{1.3} 
\caption{Number of relays sufficient to achieve certain error probabilities using uncoded transmission}
\label{tab:uncoded-number-of-required-relays}
\centering
\begin{small}
\begin{tabular}{ c | c c c c }
\hline
& $p=0.1$ & $p=0.2$ & $p=0.3$ & $p=0.4$ \\
\hline
$P_\textnormal{e}^\textnormal{up}=10^{-5}$ & $36$ & $64$ & $144$ & $576$ \\
$P_\textnormal{e}^\textnormal{up}=10^{-10}$ & $72$ & $128$ & $289$ & $1152$ \\
$P_\textnormal{e}^\textnormal{up}=10^{-50}$ & $360$ & $640$ & $1440$ & $5757$ \\
\hline
\end{tabular}
\end{small}
\end{table}

\begin{table*}[t]
\renewcommand{\arraystretch}{1.3} 
\caption{Conditions under which various coding schemes achieve the capacity}
\label{tab:comparison-schemes}
\centering
\begin{small}
\begin{tabular}{ c || c | c }
\hline
{\bf Coding schemes} & {\bf Delay} & {\bf Network size ($K$)} \\
\hline
Coded transmission with forwarding relays & Infinitely long & $K \geq K(p,\zeta)$\\
Uncoded transmission with forwarding relays & Two network uses & Infinitely large\\
Coded transmission with decoding relays & Infinitely long & $K \leq K'(p_s,p_d)$\\
\hline
\end{tabular}\\
\end{small}
\vspace*{4pt}
\vspace*{4pt}
\end{table*}

\begin{figure*}[t]
\centering
\subfloat[$p_s=p_d=0.1$]
{\resizebox{8.3cm}{!}{ 
\begin{picture}(0,0)%
\includegraphics{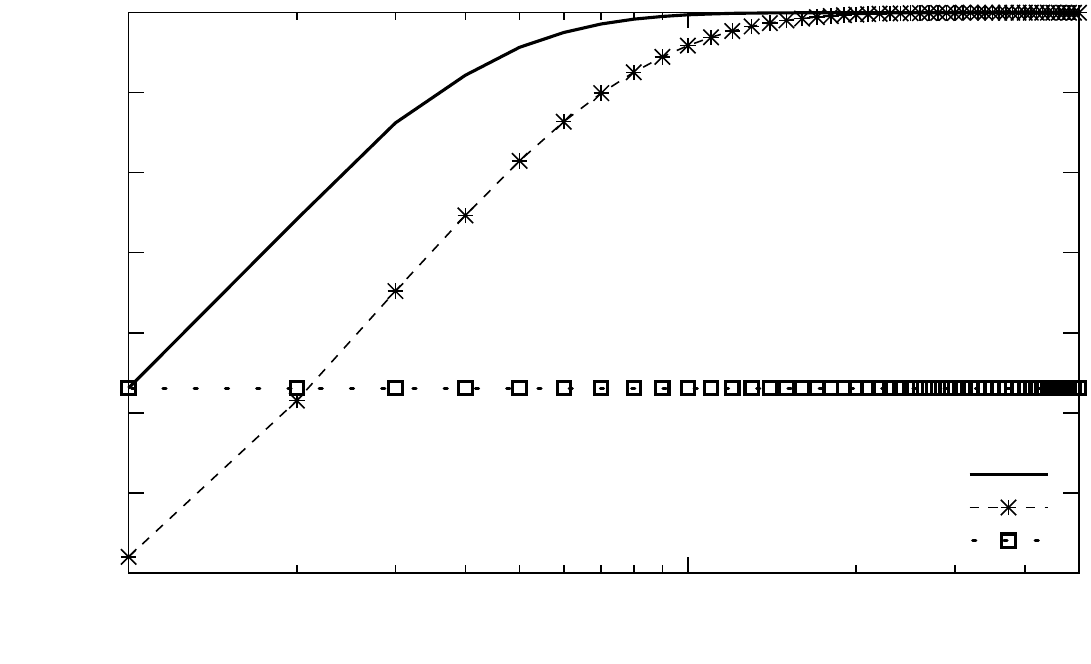}%
\end{picture}%
\setlength{\unitlength}{3947sp}%
\begingroup\makeatletter\ifx\SetFigFont\undefined%
\gdef\SetFigFont#1#2#3#4#5{%
  \reset@font\fontsize{#1}{#2pt}%
  \fontfamily{#3}\fontseries{#4}\fontshape{#5}%
  \selectfont}%
\fi\endgroup%
\begin{picture}(5233,3181)(1196,-3995)
\put(5775,-3150){\makebox(0,0)[rb]{\smash{{\SetFigFont{10}{12.0}{\familydefault}{\mddefault}{\updefault}cut-set upper bound}}}}
\put(1738,-2855){\makebox(0,0)[rb]{\smash{{\SetFigFont{10}{12.0}{\familydefault}{\mddefault}{\updefault} 0.5}}}}
\put(1738,-2471){\makebox(0,0)[rb]{\smash{{\SetFigFont{10}{12.0}{\familydefault}{\mddefault}{\updefault} 0.6}}}}
\put(1738,-2086){\makebox(0,0)[rb]{\smash{{\SetFigFont{10}{12.0}{\familydefault}{\mddefault}{\updefault} 0.7}}}}
\put(1738,-1702){\makebox(0,0)[rb]{\smash{{\SetFigFont{10}{12.0}{\familydefault}{\mddefault}{\updefault} 0.8}}}}
\put(1738,-1318){\makebox(0,0)[rb]{\smash{{\SetFigFont{10}{12.0}{\familydefault}{\mddefault}{\updefault} 0.9}}}}
\put(1738,-934){\makebox(0,0)[rb]{\smash{{\SetFigFont{10}{12.0}{\familydefault}{\mddefault}{\updefault} 1}}}}
\put(1813,-3748){\makebox(0,0)[b]{\smash{{\SetFigFont{10}{12.0}{\familydefault}{\mddefault}{\updefault} 1}}}}
\put(4498,-3748){\makebox(0,0)[b]{\smash{{\SetFigFont{10}{12.0}{\familydefault}{\mddefault}{\updefault} 10}}}}
\put(1331,-2217){\rotatebox{90.0}{\makebox(0,0)[b]{\smash{{\SetFigFont{10}{12.0}{\familydefault}{\mddefault}{\updefault}$R$ [bits/channel use]}}}}}
\put(4094,-3935){\makebox(0,0)[b]{\smash{{\SetFigFont{10}{12.0}{\familydefault}{\mddefault}{\updefault}Number of relays, $K$}}}}
\put(5775,-3309){\makebox(0,0)[rb]{\smash{{\SetFigFont{10}{12.0}{\familydefault}{\mddefault}{\updefault}$R_\text{coded,f}$ (forwarding relays)}}}}
\put(1738,-3623){\makebox(0,0)[rb]{\smash{{\SetFigFont{10}{12.0}{\familydefault}{\mddefault}{\updefault} 0.3}}}}
\put(1738,-3239){\makebox(0,0)[rb]{\smash{{\SetFigFont{10}{12.0}{\familydefault}{\mddefault}{\updefault} 0.4}}}}
\put(5775,-3468){\makebox(0,0)[rb]{\smash{{\SetFigFont{10}{12.0}{\familydefault}{\mddefault}{\updefault}$R_\text{coded,d}$ (decoding relays)}}}}
\end{picture}%
} 
\label{fig:forwarding-vs-coding}}
\hspace{0.5cm}
\subfloat[$p_s=0.05$ and $p_d=0.3$]
{\hspace{1ex}\resizebox{8.5cm}{!}{ 
\begin{picture}(0,0)%
\includegraphics{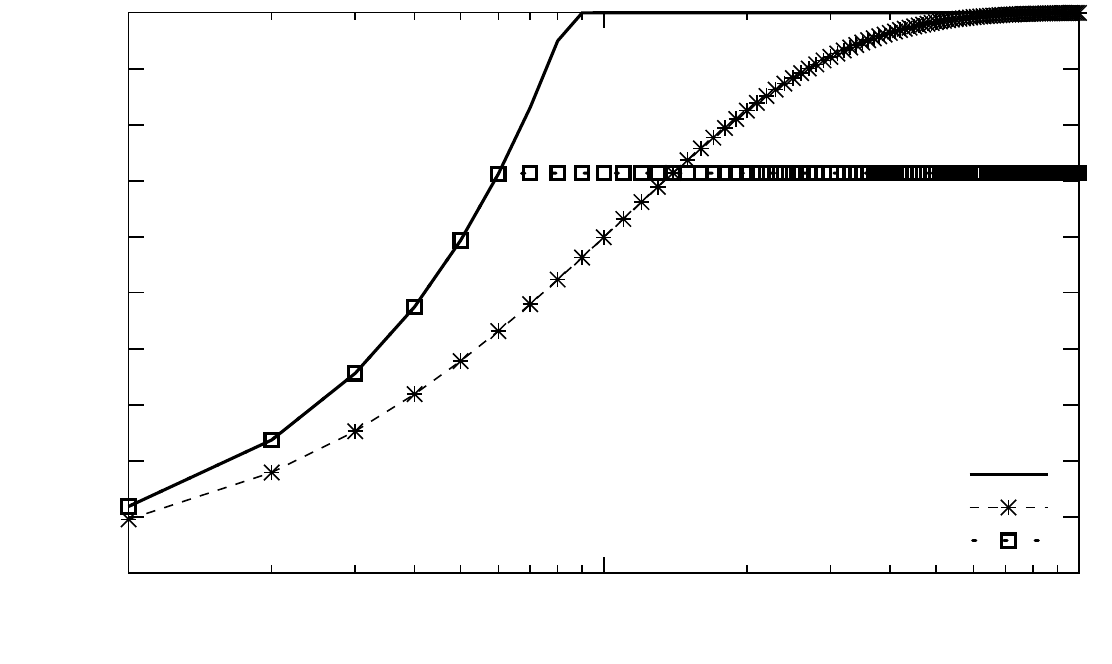}%
\end{picture}%
\setlength{\unitlength}{3947sp}%
\begingroup\makeatletter\ifx\SetFigFont\undefined%
\gdef\SetFigFont#1#2#3#4#5{%
  \reset@font\fontsize{#1}{#2pt}%
  \fontfamily{#3}\fontseries{#4}\fontshape{#5}%
  \selectfont}%
\fi\endgroup%
\begin{picture}(5324,3181)(1196,-3995)
\put(5775,-3150){\makebox(0,0)[rb]{\smash{{\SetFigFont{10}{12.0}{\familydefault}{\mddefault}{\updefault}cut-set upper bound}}}}
\put(1738,-3085){\makebox(0,0)[rb]{\smash{{\SetFigFont{10}{12.0}{\familydefault}{\mddefault}{\updefault} 0.2}}}}
\put(1738,-2816){\makebox(0,0)[rb]{\smash{{\SetFigFont{10}{12.0}{\familydefault}{\mddefault}{\updefault} 0.3}}}}
\put(1738,-2547){\makebox(0,0)[rb]{\smash{{\SetFigFont{10}{12.0}{\familydefault}{\mddefault}{\updefault} 0.4}}}}
\put(1738,-2278){\makebox(0,0)[rb]{\smash{{\SetFigFont{10}{12.0}{\familydefault}{\mddefault}{\updefault} 0.5}}}}
\put(1738,-2010){\makebox(0,0)[rb]{\smash{{\SetFigFont{10}{12.0}{\familydefault}{\mddefault}{\updefault} 0.6}}}}
\put(1738,-1741){\makebox(0,0)[rb]{\smash{{\SetFigFont{10}{12.0}{\familydefault}{\mddefault}{\updefault} 0.7}}}}
\put(1738,-1472){\makebox(0,0)[rb]{\smash{{\SetFigFont{10}{12.0}{\familydefault}{\mddefault}{\updefault} 0.8}}}}
\put(1738,-1203){\makebox(0,0)[rb]{\smash{{\SetFigFont{10}{12.0}{\familydefault}{\mddefault}{\updefault} 0.9}}}}
\put(1738,-934){\makebox(0,0)[rb]{\smash{{\SetFigFont{10}{12.0}{\familydefault}{\mddefault}{\updefault} 1}}}}
\put(1813,-3748){\makebox(0,0)[b]{\smash{{\SetFigFont{10}{12.0}{\familydefault}{\mddefault}{\updefault} 1}}}}
\put(4094,-3748){\makebox(0,0)[b]{\smash{{\SetFigFont{10}{12.0}{\familydefault}{\mddefault}{\updefault} 10}}}}
\put(6375,-3748){\makebox(0,0)[b]{\smash{{\SetFigFont{10}{12.0}{\familydefault}{\mddefault}{\updefault} 100}}}}
\put(1331,-2217){\rotatebox{90.0}{\makebox(0,0)[b]{\smash{{\SetFigFont{10}{12.0}{\familydefault}{\mddefault}{\updefault}$R$ [bits/channel use]}}}}}
\put(4094,-3935){\makebox(0,0)[b]{\smash{{\SetFigFont{10}{12.0}{\familydefault}{\mddefault}{\updefault}Number of relays, $K$}}}}
\put(5775,-3309){\makebox(0,0)[rb]{\smash{{\SetFigFont{10}{12.0}{\familydefault}{\mddefault}{\updefault}$R_\text{coded,f}$ (forwarding relays)}}}}
\put(1738,-3623){\makebox(0,0)[rb]{\smash{{\SetFigFont{10}{12.0}{\familydefault}{\mddefault}{\updefault} 0}}}}
\put(1738,-3354){\makebox(0,0)[rb]{\smash{{\SetFigFont{10}{12.0}{\familydefault}{\mddefault}{\updefault} 0.1}}}}
\put(5775,-3468){\makebox(0,0)[rb]{\smash{{\SetFigFont{10}{12.0}{\familydefault}{\mddefault}{\updefault}$R_\text{coded,d}$ (decoding relays)}}}}
\end{picture}%
} 
\label{fig:forwarding-vs-coding-2}}
\hspace{1cm}
\caption{Capacity upper bound and achievable rates (with $P_\textnormal{e} \rightarrow 0$) for coded transmission with forwarding relays and decoding relays}
\label{fig:forwarding-vs-coding-vs-cut-set}
\end{figure*}
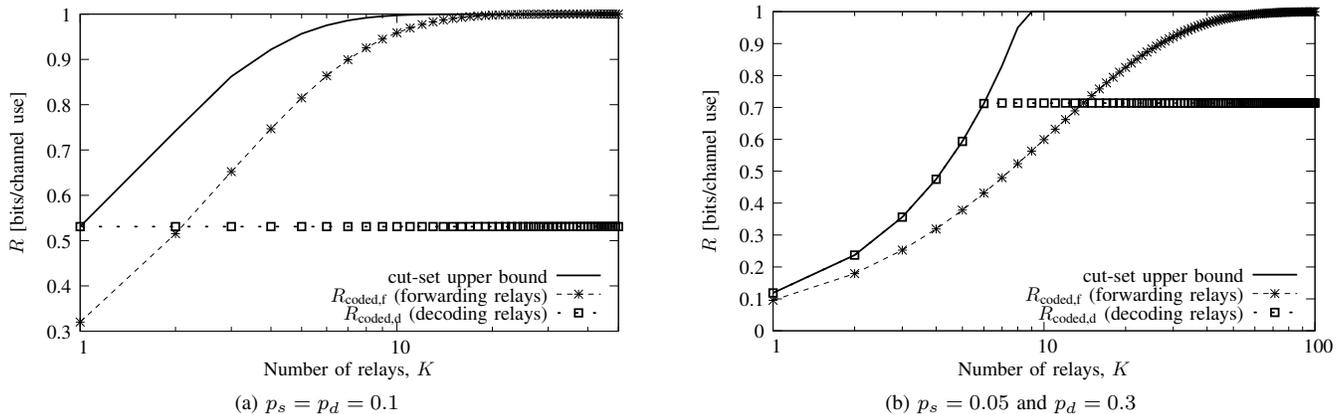

\subsection{Coded Transmission with Decoding Relays}

We have seen that forwarding relays are asymptotically optimal for large $K$.
Now, we show that coded transmission with decoding relays achieves the capacity when $K$ is smaller than a certain positive integer which depends on the cross-over probabilities. From Theorem~\ref{theorem:finite-k-capacity}, we have the following result.
\begin{corollary} \label{corollary:df}
Consider the symmetrical BSPR network.  If $K \leq K'(p_s,p_d)$, where
\begin{equation} \label{eq:kdash}
K'(p_s,p_d) = \max \left\{1,\frac{1-H(p_s)}{1-H(p_d)} \right\},
\end{equation}
then coded transmission with decoding relays achieves the capacity.
\end{corollary}

The proof of $R_\textnormal{coded,d}$ again relies on the channel coding theorem. Hence, we also have an infinitely long delay for coded transmission with decoding relays. 


\begin{remark} \label{remark:suboptimal}
From Theorem~\ref{lemma:upper-bound-full-decoding}, we know that for any coding scheme in which one or more relays are to decode the source message, rates at most $R_\textnormal{decode}=1-H(p_s)$ can be achieved, which is independent of $K$ and is bounded away from the capacity upper bound of 1~bit/network use, for any $p_s > 0$. Since we have shown that as $K$ increases, rates arbitrarily close to 1~bit/channel are achievable, it follows that there exists a positive integer $K'$ where for any $K \geq K'$, we have that $R_\textnormal{decode} < R_\textnormal{coded,f} < C$, i.e., any coding scheme that requires decoding of the source message at any relay is suboptimal in a large BSPR network---in such cases, forwarding relays achieve strictly higher rates.
\end{remark}

\subsection{Summary}

From the above subsections, the three coding schemes  achieve the capacity of the BSPR network under different conditions as summarized in Table~\ref{tab:comparison-schemes}.

\subsection{Numerical Examples}

We present two numerical examples to compare achievable rates of coded transmission with forwarding relays and decoding relays to the capacity upper bound with varying $K$ and with the following parameter values: (i) $p_s= p_d=0.1$ and (ii) $p_s= 0.05$, $p_d=0.3$. In the first network, the source-to-relay and the relay-to-destination channels are equally noisy, while in the second network, the source-to-relay channels are less noisy.

The results are shown in Figs.~\ref{fig:forwarding-vs-coding} and \ref{fig:forwarding-vs-coding-2}. We see that decoding relays achieve the capacity when $K$ is small, i.e.,  $K=1$ when $p_s=p_d=0.1$, and $K \leq 6$ when $p_s=0.05$ and $p_d=0.3$. We can see from \eqref{eq:df-symmetry} that when $K \geq \frac{1-H(p_s)}{1-H(p_d)}$, the maximum achievable rate of decoding relays is fixed at $R_\textnormal{coded,d} = 1-H(p_s)$. Hence, using decoding relays is suboptimal when the number of relays is large (where forwarding relays achieve close to 1~bit/network use). Using forwarding relays, as predicted, approaches the capacity upper bound as the number of relays increases.

\begin{figure*}[!t]
\small
\begin{subequations}
\begin{align}
&p^*(v_1,v_2,\dotsc,v_K|u) =\begin{bmatrix}
p^*(00 \dotsm 0|0) & p^*(00 \dotsm 1|0) & \dotsm &p^*(11 \dotsm 1|0)\\
p^*(00 \dotsm 0|1) & p^*(00 \dotsm 1|1) & \dotsm &p^*(11 \dotsm 1|1)
\end{bmatrix}\label{eq:matrix}\\
&= \left[\;
\boxed{\begin{matrix}
p^*(00 \dotsm 0|0) & p^*(11 \dotsm 1|0)\\
p^*(00 \dotsm 0|1) & p^*(11 \dotsm 1|1)
\end{matrix}}\;
\boxed{\begin{matrix}
p^*(00 \dotsm 1|0) & p^*(11 \dotsm 0|0)\\
p^*(00 \dotsm 1|1) & p^*(11 \dotsm 0|1)
\end{matrix}}\;\dotsm\;
\boxed{\begin{matrix}
p^*(10001 \dotsm 01|0) & p^*(01110 \dotsm 10|0)\\
p^*(10001 \dotsm 01|1) & p^*(01110 \dotsm 10|1)
\end{matrix}}\;\dotsm\;
\right]\label{eq:matrix-grouped}
\end{align}
\end{subequations}
\hrulefill
\vspace*{4pt}
\end{figure*}

\section{Reflection}\label{sec:reflection}

We have investigated the binary-symmetric parallel-relay network. We derived achievable rates using different coding schemes that utilize forwarding or decoding relays or a mix thereof, as well as coded or uncoded transmission at the source. We have also analyzed the network as the number of relays grows to infinity.

With coded transmission, forwarding relays achieve the capacity (rounded to some number of significant figures) for networks with a finite number of relays. For instance, for cross-over probability $p=0.1$, we need 16 relays or more to achieve within 0.0001 bits of the capacity, and for $p=0.4$, we need 387 relays or more to achieve the same result. However, an infinitely long code length, $n$, is required to drive the error probability to zero. Decoding is done after the destination receives $n$ channel outputs over time which necessarily incurs a large delay.

With coded transmission, decoding relays achieve the capacity of networks with one relay, and networks with more relays if the sum of capacities of all the channels from the relays to the destination is smaller than the capacity of the channel from the source to each relay. 
Again, decoding is done after the destination receives $n$ channel outputs over time, which necessarily incurs a large delay.

With uncoded transmission, decoding is almost ``instantaneous'', i.e., when the destination receives the noisy bit transmitted by the source. The transmission from the source to the relays and the transmissions from the relays to the destination take two network uses, and this is the total delay incurred. With uncoded transmission, message bits are sent at 1~bits/network use which is an upper bound to the capacity. However, to drive the error probability to zero, an infinitely large number of relays is required.

Coded transmission with decoding relays (which removes the noise on the source-to-relays channels) performs well when the number of relays is small because the number of additional rate constraints---required since each relay must decode the source message---is also small.
These additional rate constraints actually limit its performance when the number of relays increases. Using forwarding relays, even though the noise on the source-to-relays channels propagates to the relays-to-destination channels, the increase in the number of relays provides the destination with sufficient spatial diversity to decode the source message. We can view both coded and uncoded transmission schemes with forwarding relays as ``spatial'' repetition codes, where the minimum Hamming distance of the code increases with the codelength, and the codelength increases with the number of relays.

The aforementioned observations lead to the design of the hybrid scheme, where we use some relays as forwarding relays and the others as decoding relays. More specifically, we use relays with better source-to-relay channels as decoding relays, and relays with noisy source-to-relay channels as forwarding relays. The reason for this choice is that if a relay is able to decode the source message without constraining the overall transmission rate (which is likely when the channel from the source to {\em this} relay is good), we should let the relay decode the source message on that link to stop noise propagation.
This scheme can improve the performance over purely forwarding relays and purely decoding relays in non-symmetrical BSPR networks where the source-to-relay channels are not all equally noisy.

\appendices


\section{Proof of Theorem~\ref{theorem:cut-set-upper-bound}} \label{appendix:ub}
First, we note that $W \rightarrow \boldsymbol{U} \rightarrow \bar{\boldsymbol{V}} \rightarrow \bar{\boldsymbol{X}} \rightarrow \bar{\boldsymbol{Y}} \rightarrow \hat{W}$ forms a Markov chain. Using Fano's inequality~\cite[Lem.\ 7.9.1]{coverthomas06}, we have that $H(W|\hat{W}) \leq 1 + P_\textnormal{e} nR$.
So,
\begin{subequations}
\begin{align}
nR &= H(W) = H(W|\hat{W}) + I(W;\hat{W}) \\
&\leq 1 + P_\textnormal{e} nR + I(\boldsymbol{U};\bar{\boldsymbol{Y}}) \label{eq:fano-2} \\
R & \leq I(U;\bar{V}) + P_\textnormal{e} R + 1/n \label{eq:upper-bound-1}
\end{align}
\end{subequations}
where \eqref{eq:fano-2} follows from Fano's inequality and by applying the data processing inequality (DPI)~\cite[Thm.\ 2.8.1]{coverthomas06} to the aforementioned Markov chain, and \eqref{eq:upper-bound-1} follows by applying DPI again and because the channel from $U$ to $\bar{V}$ is memoryless.

We also have that
\begin{subequations}
\begin{align}
&I(\boldsymbol{U};\bar{\boldsymbol{Y}}) \leq I(\bar{\boldsymbol{X}};\bar{\boldsymbol{Y}}) \leq n I(\bar{X};\bar{Y}) \label{eq:memoryless-a}\\
&= n \sum_{i=1}^K I(X_i;Y_i) \leq n \sum_{i=1}^K [1 - H(p_{i,d})] \label{eq:uniform-max}
\end{align}
\end{subequations}
where \eqref{eq:memoryless-a} follows from the DPI and  because the channels from $\bar{X}$ to $\bar{Y}$ are memoryless, and \eqref{eq:uniform-max} is derived because the uniform distribution $p'(x_i)$ maximizes the mutual information.
From \eqref{eq:fano-2} and \eqref{eq:uniform-max}, we have
\begin{equation}
R \leq K - \sum_{i=1}^K H(p_{i,d}) + P_\textnormal{e} R + 1/n. \label{eq:upper-bound-2}
\end{equation}

Setting $n \rightarrow \infty$ and then $P_\textnormal{e} \rightarrow 0$ for conditions \eqref{eq:upper-bound-1} and \eqref{eq:upper-bound-2}, we have Theorem~\ref{theorem:cut-set-upper-bound}. \hfill $\blacksquare$

\section{Proof of Lemma~\ref{lemma:symmetric}} \label{appendix:symmetric}

We write the matrix of transition probabilities of the channel as \eqref{eq:matrix}. The top row is the conditional probability of $\bar{V}$ given $U=0$ and the bottom row $U=1$. We rearrange the columns and pair up columns $p^*(\bar{v}|u)$ and $p^*(\bar{\bar{v}}|u)$ as a \emph{sub-matrix}, where $\bar{\bar{v}}$ is $\bar{v}$ with all the bits flipped. Each sub-matrix is boxed in \eqref{eq:matrix-grouped}. Clearly, $p^*(\bar{v}|u')=p^*(\bar{\bar{v}}|u'')$ if $u' \neq u''$. In each sub-matrix, the top row is a permutation of the bottom row, and the left column is a permutation of the right column. Hence, the channel is symmetric in the sense of \cite[page 94]{gallager68}. \hfill $\blacksquare$

\section{Proof of Theorem~\ref{thm:uncoded}} \label{appendix:uncoded-symmetrical}

As forwarding relays are used, we have the equivalent channel in Fig.~\ref{fig:equivalent-channel-simple-relays} with $U=W$.
At the destination, the received signals are $Y_i = W \oplus N_i$, where $N_i = Z_i \oplus E_i$ and $\Pr\{N_i=1\}=p$ as defined in \eqref{eq:equivalent-cross-over}.

Let $\bar{y} = (y_1,y_2,\dotsc,y_K)$ be the received signals at the destination. The optimal decision decoding rule, which minimizes the error probability, is:
\begin{equation}
\hat{w} = \begin{cases}0, & \textnormal{if } \Pr\{W=0|\bar{Y}=\bar{y}\} \geq \Pr\{W=1|\bar{Y}=\bar{y}\}\\
1, & \textnormal{otherwise}
\end{cases}.\label{eq:estimation-1}
\end{equation}
Applying Bayes' rule to \eqref{eq:estimation-1} and noting that $\Pr\{W=0\} = \Pr\{W=1\}=\frac{1}{2}$, the condition for $\hat{w}=0$ in \eqref{eq:estimation-1} becomes
\begin{equation}
\Pr\{\bar{Y}=\bar{y}|W=0\} \geq \Pr\{\bar{Y}=\bar{y}|W=1\} \label{eq:map}.
\end{equation}
This is known as the maximum likelihood decoder. Since $p_i=p$, $\forall i \in \{1,2,\dotsc,K\}$. The decision rule in \eqref{eq:map} becomes
\begin{equation}
(1-p)^{\mathbf{0}(\bar{y})}p^{K-\mathbf{0}(\bar{y})} \geq (1-p)^{K-\mathbf{0}(\bar{y})} p^{\mathbf{0}(\bar{y})},\label{eq:estimation-2}
\end{equation}
where $\mathbf{0}(\bar{y})$ is the number of 0's in $\bar{y}$. As $0 \leq p \leq \frac{1}{2}$, we have the optimal decoding function at the decoder as follows.
\begin{equation}
\hat{W} = g(\bar{Y}) = \begin{cases}
0, &\textnormal{if } \mathbf{0}(\bar{Y}) \geq \frac{K}{2}\\
1, &\textnormal{otherwise}
\end{cases}.\label{eq:uncoded-decoding-rule}
\end{equation}

Having defined the encoding and decoding functions for uncoded transmissions with forwarding relays, we derive the error probability.
\begin{align}
P_\textnormal{e} 
&= \Pr\{\hat{W}=1 | W=0\} \Pr\{W=0\}\nonumber\\
&\quad + \Pr\{\hat{W}=0 | W=1\} \Pr\{W=1\} \nonumber \\ 
&= \frac{1}{2}\Pr\left\{\mathbf{0}(\bar{Y}) < \frac{K}{2}\Big|W=0\right\}\nonumber\\
&\quad + \frac{1}{2}\Pr\left\{\mathbf{0}(\bar{Y}) \geq \frac{K}{2}\Big|W=1\right\} \nonumber \\ 
&= \Pr\left\{\mathbf{0}(\bar{Y}) \geq \frac{K}{2}\Big|W=1\right\}= \Pr\left\{\sum_{i=1}^K N_i \geq \frac{K}{2} \right\} \nonumber \\
&= \Pr\left\{\frac{1}{K}\sum_{i=1}^K N_i - p \geq \frac{1}{2} - p \right\} \nonumber \\
&\leq \exp \left[ -2K\left(\frac{1}{2}-p\right)^2\right] \triangleq P_\textnormal{e}^\textnormal{up},\label{eq:uncoded-error}
\end{align}
where \eqref{eq:uncoded-error} is due to Hoeffding~\cite[Thm.\ 2]{hoeffding63} if $\frac{1}{2} - p > 0$.

By sending uncoded bits at the rate $\frac{n}{n+1}$~bit/network use, we know from \eqref{eq:uncoded-error} that the error probability can be bounded by
\begin{equation}
P_\textnormal{e} \leq \exp \left( -K \delta \right),
\end{equation}
where $\delta = 2\left(\frac{1}{2}-p\right)^2 > 0$. So, for any $0 \leq p < \frac{1}{2}$ and  $\epsilon > 0$, we can select $K \geq \frac{1}{\delta}\ln \frac{1}{\epsilon}$ such that $P_\textnormal{e} \leq \epsilon$. \hfill $\blacksquare$

\section{Proof of Corollary~\ref{theorem:uncoded-general}} \label{appendix:uncoded-general}
We use the idea in the proof of Corollary~\ref{theorem:general-coded-forwarding-rate}, i.e., using only $M$ relays for some $M \leq K$. If a coding scheme achieves $P_\textnormal{e} \leq \epsilon$ using only $M$ relays, then it can also achieve $P_\textnormal{e} \leq \epsilon$ with $K$ relays. Again, for each sub-channel from $U$ to $Y_{m_i}$ with cross-over probability $p_{m_i}$, we further add random noise to get $Y'_{m_i} = Y_{m_i} \oplus E_{m_i}''$, where $E_{m_i}'' \in \{0,1\}$ and $\Pr\{E_{m_i}''=1\} = \frac{p_{\textnormal{MAX}(\mathcal{M})}-p_{m_i}}{1-2p_{m_i}}$. Now, each sub-channel from $U$ to $Y'_{m_i}$ is a binary symmetric channel with cross-over probability $p_{\textnormal{MAX}(\mathcal{M})}$. We use the following decoding rule:
\begin{equation}
\hat{W}  = \begin{cases}
0, &\textnormal{if } \mathbf{0}(Y_{m_1}',Y_{m_2}',\dotsc,Y_{m_M}') \geq \frac{M}{2}\\
1, &\textnormal{otherwise}
\end{cases}.
\end{equation}
From \eqref{eq:uncoded-error}, we know that the error probability of this decoder is
\begin{equation}
P_\textnormal{e} \leq \exp \left[-2 M\left(\frac{1}{2} - p_{\textnormal{MAX}(\mathcal{M})} \right)^2 \right].
\end{equation}
As $K \rightarrow \infty$, if $M\left(\frac{1}{2} - p_{\textnormal{MAX}(\mathcal{M})} \right)^2 \rightarrow \infty$, then for any $\epsilon > 0$, we can find some $\mathcal{M}$ such that $P_\textnormal{e} \leq \epsilon$. \hfill $\blacksquare$

\section{Proof of Theorem~\ref{theorem:coded-relay}} \label{appendix:coded}
We use the following {\em super-block} coding scheme. Consider $B$ blocks each consisting of $n$ network uses. We split the source message into $(B-1)$ equal parts, i.e., $W = (W_1,W_2,\dotsc, W_{B-1})$ where each $W_i$ is independent and uniformly distributed in $\{1,\dotsc, 2^{nR}\}$. In each block $b \in \{1,2,\dotsc, B-1\}$, the source transmit $\boldsymbol{U}^{(b)}(W_b)$. Denoting the received symbols of relay $i$ in block $b$ by $\boldsymbol{V}_i^{(b)}$, relay $i$ transmits $\boldsymbol{X}_i^{(b+1)}(\boldsymbol{V}_i^{(b)})$ in block $(b+1)$. The destination then decode $W_b$ from its received symbols in block $(b+1)$, i.e., $\bar{\boldsymbol{Y}}^{(b+1)}$.

In the following, we will only consider the transmissions from the source to the relay in block 1, and those from the relay to the destination in block 2. Suppose that the destination can reliably (with arbitrary low error probability) decode $W_1$. Repeating the same transmission scheme for the source--relay channel in block $b$ and for the relay--destination channel in block $(b+1)$ for all $b \in \{2,3,\dotsc, B-1\}$, the rate of $R'=(B-1)R/B$ is achievable. This means we can achieve rate $R' \rightarrow R$ by letting $B \rightarrow \infty$. Where appropriate, we drop the superscript that indicates the block to simplify notation.

We first select a set of $M$ relays, and let the set of selected relays be $\{m_1,m_2,\dotsc,m_M\} \triangleq \mathcal{M}$.
In block 1, the source sends $\boldsymbol{U}(W_1)$. Choosing $p'(u)$ to be the uniform distribution, if
\begin{equation}
R < I(U;V_{m_i}) = 1 - H(p_{s,m_i}), \label{eq:decode-1}
\end{equation}
for all $m_i \in \mathcal{M}$, then each relay in $\mathcal{M}$ can reliably (with arbitrarily low error probability when $n$ is sufficiently large) decode $W_1$.

In block 2, the relays in $\mathcal{M}$ transmit the decoded $W_1$ to the destination. The rest of the relays transmit zero, $X_j=0$ for all $j \notin \mathcal{M}$. Choosing $p(x_{m_1}, x_{m_2}, \dotsc, x_{m_M}) = \prod_{i=1}^M p'(x_{m_i})$, where each $p'(x_{m_i})$ is the uniform distribution, if
\begin{subequations}
\begin{align}
R &< I(X_{m_1}, \dotsc, X_{m_M} ; Y_{m_1}, \dotsc, Y_{m_M})\\
&= \sum_{i=1}^M I(X_{m_i};Y_{m_i}) = \sum_{i=1}^M [ 1 - H(p_{i,d})], \label{eq:decode-2}
\end{align}
\end{subequations}
then the destination can reliably decode $W_1$. Note that $X_{m_i}$ are independently generated. So, \eqref{eq:decode-2} follows because $Y_{m_i}$ is independent of $\{Y_{m_j}: j \neq i\}$, and $Y_{m_i}$ is also independent of $\{X_{m_j}: j \neq i\}$ given $X_{m_i}$.

If a rate $R$ satisfies \eqref{eq:decode-1} for all $m_i \in \mathcal{M}$ and \eqref{eq:decode-2}, then the destination can reliably decode $W_1$. Repeating this scheme for all blocks, we have Theorem~\ref{theorem:coded-relay}. \hfill $\blacksquare$

\section{Proof of Theorem~\ref{theorem:finite-k-capacity}} \label{appendix:finite-k-capacity}
Compare $R_\textnormal{ub}$ in Theorem~\ref{theorem:cut-set-upper-bound} and $R_\textnormal{coded,d}$ in Theorem~\ref{theorem:coded-relay}. For $K=1$, we set $M=1$, and we have $\max_{p(u)}I(U;V_1) = 1-H(p_{s,1})$. So, $R_\textnormal{coded,d}= R_\textnormal{ub}$, and is the capacity.

For $K>1$, setting $M=K$, if $K - \sum_{j=1}^K H(p_{j,d}) \leq 1 - H(p_{s,i})$ for all $i \in \{1,2,\dotsc, K\}$, then $R_\textnormal{coded,d}$ reduces to \eqref{decoding-relay-capacity}.
Since $I(U;V_1,V_2\dotsc,V_K) \geq I(U;V_i)$
for all $i$, we also have $\max_{p(u)} I(U;\bar{V}) \geq K - \sum_{j=1}^K H(p_{j,d})$,
and $R_\textnormal{ub}$ also reduces to \eqref{decoding-relay-capacity}. \hfill $\blacksquare$

\section{Proof of Theorem~\ref{theorem:coded-forward-asymptotic}} \label{appendix:coded-forward-asymptotic}

We first assume that $U=0$ is sent. We define the $\eta$-strongly typical set, denoted by $\mathcal{T}_{[Y|0]\eta}^K$, with respect to the distribution $p^*(y|u=0)$ as a set of vectors $\{ \bar{y} \}$ such that
\begin{align}
\left\vert \frac{1}{K}N(1;\bar{y}) - p \right\vert + \left\vert \frac{1}{K}N(0;\bar{y}) - q \right\vert < \eta,
\end{align}
where $N(a;\bar{y})$ is the number of occurrences of the symbol $a \in \{0,1\}$ in the sequence $\bar{y}$. See \cite[page 73]{yeung02} for a more general definition of strongly typical sets.

We assume that $0 < p < \frac{1}{2}$ is a rational number and consider a sufficiently large integer $K'$ such that $K'p$ is an integer. For any $K'$ and $p$, we can choose $\eta$ sufficiently small such that $\mathcal{T}_{[Y|0]\eta}^{K'} =$ \{all  $\bar{y}$ each having $K'p$ number of 1's and $K'q$ number of 0's\}. This means, for each $\bar{y} \in \mathcal{T}_{[Y|0]\eta}^{K'}$, we have $p^*(\bar{y}|0) = p^{K'p}q^{K'q}$. By the same argument, $\mathcal{T}_{[Y|1]\eta}^{K'} =$ \{all $\bar{y}$ each having  $K'p$ number of 0's and $K'q$ number of 1's \}, and hence $p^*(\bar{y}|1) = p^{K'p}q^{K'q}$.

Recall that the input distribution $p'(u)$ is the uniform distribution. Suppose that some $\bar{y}$ has $K'p$ number of 0's and $K'q$ number of 1's. We have $p^*(\bar{y}|0) = q^{K'p}p^{K'q}$ and $p^*(\bar{y}|1) = q^{K'q}p^{K'p}$. So,
\begin{align}
p^*(\bar{y}) &= \Pr\{U=0\} p^*(\bar{y}_{(l)}|0) + \Pr\{U=1\} p^*(\bar{y}_{(l)}|1)\nonumber \\
&= \frac{1}{2} (q^{K'p}p^{K'q} + q^{K'q}p^{K'p}) \triangleq \alpha(K',p). \label{eq:y-bar}
\end{align}

Since the effective noise, $N_i$, in each channel $U \rightarrow Y_i$ is independent and $\Pr\{N_i=1\}$ is the same for all $i \in \{1,2,\dotsc,K\}$, $Y_i$ are i.i.d. given $U$ [from \eqref{eq:equivalent-channel-2}]. So, we have, for $K'$ sufficiently large, \cite[Theorem 5.2]{yeung02}
\begin{equation}
\Pr \{ \bar{Y} \in  \mathcal{T}_{[Y|U]\eta}^{K'} | U\} > 1 - \eta.
\end{equation}

Now, for some sufficiently large $K'$,
\begin{align*}
&H(\bar{Y}) = - \sum_{\bar{y}} p^*(\bar{y}) \log p^*(\bar{y}) \nonumber\\
&= -\sum_\textnormal{ all $\bar{y}$ that have $K'p$ number of 0's} p^*(\bar{y}) \log p^*(\bar{y}) \nonumber \\ &\quad - \sum_\textnormal{ all $\bar{y}$ that have $K'q$ number of 0's} p^*(\bar{y}) \log p^*(\bar{y})\nonumber\\
&\quad - \sum_{\textnormal{ all $\bar{y}$ that do not have $K'p$ or $K'q$ number of 0's}} p^*(\bar{y}) \log p^*(\bar{y})   \nonumber\\
&> -\sum_{\bar{y} \in \mathcal{T}_{[Y|1]\eta}^{K'} } p^*(\bar{y}) \log p^*(\bar{y}) - \sum_{\bar{y} \in \mathcal{T}_{[Y|0]\eta}^{K'}} p^*(\bar{y}) \log p^*(\bar{y}) \nonumber\\
&= -\sum_{\bar{y} \in \mathcal{T}_{[Y|1]\eta}^{K'} } \frac{1}{2}[p^*(\bar{y}|1) + p^*(\bar{y}|0)] \log p^*(\bar{y}) \nonumber\\
&\quad- \sum_{\bar{y} \in \mathcal{T}_{[Y|0]\eta}^{K'}} \frac{1}{2}[p^*(\bar{y}|1) + p^*(\bar{y}|0)] \log p^*(\bar{y}) \nonumber\\
&> -\sum_{\bar{y} \in \mathcal{T}_{[Y|1]\eta}^{K'} } \frac{1}{2}p^*(\bar{y}|1) \log p^*(\bar{y}) \nonumber \\
&\quad - \sum_{\bar{y} \in \mathcal{T}_{[Y|0]\eta}^{K'}} \frac{1}{2} p^*(\bar{y}|0) \log p^*(\bar{y}) \nonumber\\
& \stackrel{(\sharp)}=  - \frac{1}{2} \log\alpha(K',p) \sum_{\bar{y} \in \mathcal{T}_{[Y|1]\eta}^{K'} } p^*(\bar{y}|1) \nonumber\\
&\quad - \frac{1}{2} \log\alpha(K',p)\sum_{\bar{y} \in \mathcal{T}_{[Y|0]\eta}^{K'} } p^*(\bar{y}|0)\\
&=  - \frac{1}{2} \log\alpha(K',p) \Pr\{\bar{Y} \in \mathcal{T}_{[Y|1]\eta}^{K'}|U=1\}\nonumber\\
&\quad - \frac{1}{2} \log\alpha(K',p) \Pr\{\bar{Y} \in \mathcal{T}_{[Y|0]\eta}^{K'}|U=0\} \nonumber\\
&> - \log\alpha(K',p) (1 - \eta) \nonumber\\
& = -\log\alpha(K',p) - \eta',
\end{align*}
where $\eta' = -\eta\log\alpha(K',p)>0$ which can be chosen arbitrarily small by choosing an arbitrarily small $\eta$ for any $K'$ and $p$. Note that the equality $(\sharp)$ follows from \eqref{eq:y-bar}.

Now,
\begin{subequations}
\begin{align}
&R_\textnormal{coded,f} = I(U;\bar{Y}) = H(\bar{Y}) - H(\bar{Y}|U)\\
&> -\log\alpha(K',p) - \eta'- K'H(p)\\
&= -\log\alpha(K',p) + K' (p \log p + q \log q) - \eta'\\
&= -\log \left( \frac{2}{1 + \left(\frac{p}{1-p}\right)^{K'(1-2p)}} \right) - \eta' \label{eq:coded-asymptotic}\\
& \rightarrow 1 = R_\textnormal{ub}, \quad \textnormal{as} \quad K' \rightarrow \infty, \label{eq:proof-asymptotic}
\end{align}
\end{subequations}
where $\eta'$ can be chosen arbitrarily small for any $K'$ and $p$, and 1 bit/network use is an upper bound to the capacity.
Note that $\frac{p}{1-p} < 1$, and $1 - 2p>0$. The above result holds for large $K'$ where $K'p$ is an integer, But since $I(U;Y_1\dotsc, Y_{K+1}) \geq I(U;Y_1,\dotsc, Y_K)$ for all positive integers $K$, we have Theorem~\ref{theorem:coded-forward-asymptotic}. \hfill $\blacksquare$


\end{document}